\documentclass[11pt]{article}

\usepackage[utf8]{inputenc} 
\usepackage[T1]{fontenc}    
\usepackage{hyperref}       
\usepackage{url}            
\usepackage{booktabs}       
\usepackage{amsfonts}       
\usepackage{nicefrac}       
\usepackage{microtype}      
\usepackage{xcolor}         

\usepackage{complexity}

\usepackage{hhline} 
\usepackage[font=small,labelfont=bf]{caption}

\usepackage{authblk}

\usepackage{mathtools}
\usepackage{amsmath}
\usepackage{amssymb}
\usepackage{amsthm}

\usepackage{bbm}

\usepackage{tikz}
\usepackage{colortbl}

\definecolor{darkgray}{gray}{0.85} 

\usepackage{mleftright}
\usepackage{thm-restate}
\usepackage{bm}

\newcommand{\declarecolor}[2]{\definecolor{#1}{RGB}{#2}\expandafter\newcommand\csname #1\endcsname[1]{\textcolor{#1}{##1}}}
\declarecolor{White}{255, 255, 255}
\declarecolor{Black}{0, 0, 0}
\declarecolor{Maroon}{128, 0, 0}
\declarecolor{Coral}{255, 127, 80}
\declarecolor{Red}{182, 21, 21}
\declarecolor{LimeGreen}{50, 205, 50}
\declarecolor{DarkGreen}{0, 80, 0}
\declarecolor{Purple}{146, 42, 158}
\declarecolor{Navy}{0, 0, 128}
\declarecolor{LightBlue}{84, 101, 202}
\definecolor{mydarkblue}{rgb}{0,0.08,0.45}
\hypersetup{ %
    pdftitle={},
    pdfkeywords={},
    pdfborder=0 0 0,
    pdfpagemode=UseNone,
    colorlinks=true,
    linkcolor=Navy,
    citecolor=DarkGreen,
    filecolor=Purple,
    urlcolor=Purple,
}

\usepackage{fullpage}

\usepackage[ruled,linesnumbered]{algorithm2e}

\SetKwComment{Comment}{/* }{ */}

\usepackage[capitalize,noabbrev,nameinlink]{cleveref}

\usepackage{MnSymbol}

\usepackage[natbib,maxcitenames=6,maxnames=6,style=alphabetic]{biblatex}

\bibliography{main}

\let\E\relax

\usepackage{enumitem}

\let\log\relax
\let\poly\relax

\newcommand{\log}{\mathsf{log}}
\newcommand{\poly}{\mathsf{poly}}

\newcommand{\dens}{\mathrm{dens}}

\newcommand{\supp}{\text{supp}}

\newcommand{\NPtilde}{\mathsf{N}\widetilde{\mathsf{P}}}

\newcommand{\sparsecce}{\textsc{OptimalSparseCCE}}

\newcommand{\unsparsecce}{\textsc{SparseCCE}}

\newcommand{\uniquesparsecce}{\textsc{UniqueUniformSparseCCE}}

\newcommand{\subsetsparsecce}{\textsc{SubsetUniformSparseCCE}}

\newcommand{\newaction}{\mathsf{O}}

\newcommand{\hatvx}{\widehat{\vec{x}}}

\newcommand{\hatvy}{\widehat{\vec{y}}}

\newcommand{\epseq}{\epsilon}
\newcommand{\epswel}{\hat{\epsilon}}

\newcommand{\opt}{\mathsf{OPT}}

\newcommand{\uni}{\vec{u}}

\newcommand{\maxclique}{\textsc{MaxClique}}
\newcommand{\plantedclique}{\textsc{PlantedClique}}

\newcommand{\TO}{T_{\newaction}}
\newcommand{\barTO}{\overline{T}_{\newaction}}

\newcommand{\hatvmu}{\widehat{\vmu}}

\DeclareMathOperator{\pr}{\mathbb{P}}

\newcommand{\reg}{\mathsf{Reg}}

\newcommand{\vu}{\vec{u}}

\newcommand*{\Q}{{\mathbb{Q}}}

\newcommand*{\N}{{\mathbb{N}}}

\let\R\relax
\newcommand*{\R}{{\mathbb{R}}}
\newcommand*{\E}{{\mathbb{E}}}

\newcommand*{\cG}{{\mathcal{G}}}

\newcommand*{\cA}{{\mathcal{A}}}

\newcommand{\defeq}{\coloneqq}

\newcommand{\nakedcite}[1]{\citeauthor{#1}, \citeyear{#1}}

\newcommand{\vx}{\vec{x}}
\newcommand{\vy}{\vec{y}}

\newcommand{\vmu}{\vec{\mu}}

\DeclareMathOperator{\sw}{\mathsf{SW}}

\newcommand{\vxstar}{\vec{x}^\star}

\renewcommand{\vec}[1]{\bm{#1}}
\newcommand{\mat}[1]{\mathbf{#1}}

\theoremstyle{plain}
\newtheorem{theorem}{Theorem}[section]
\newtheorem{lemma}[theorem]{Lemma}
\newtheorem{corollary}[theorem]{Corollary}
\newtheorem{proposition}[theorem]{Proposition}

\newtheorem{conjecture}[theorem]{Conjecture}

\newtheorem{claim}[theorem]{Claim}

\theoremstyle{definition}
\newtheorem{definition}[theorem]{Definition}

\theoremstyle{remark}

\title{Barriers to Welfare Maximization with\\ No-Regret Learning}

\author[1]{Ioannis Anagnostides}
\author[2]{Alkis Kalavasis}
\author[3]{Tuomas Sandholm}

\affil[1,3]{Carnegie Mellon University}
\affil[2]{Yale University}
\affil[3]{Strategy Robot, Inc.}
\affil[3]{Strategic Machine, Inc.}
\affil[3]{Optimized Markets, Inc.}
\affil[ ]{\texttt{\{ianagnos,sandholm\}}\texttt{@cs.cmu.edu}, \texttt{alkis.kalavasis}\texttt{@yale.edu}}

\begin{document}

\maketitle

\begin{abstract}
    A celebrated result in the interface of online learning and game theory guarantees that the repeated interaction of \emph{no-regret} players leads to a \emph{coarse correlated equilibrium (CCE)}---a natural game-theoretic solution concept. Despite the rich history of this foundational problem and the tremendous interest it has received in recent years, a basic question still remains open: how many iterations are needed for no-regret players to approximate an equilibrium? In this paper, we establish the first \emph{computational} lower bounds for that problem in two-player (general-sum) games under the constraint that the CCE reached approximates the optimal social welfare (or some other natural objective). From a technical standpoint, our approach revolves around proving lower bounds for computing a near-optimal \emph{$T$-sparse} CCE---a mixture of $T$ product distributions, thereby circumscribing the iteration complexity of no-regret learning even in the centralized model of computation. Our proof proceeds by extending a classical reduction of Gilboa and Zemel [1989] for optimal Nash to sparse (approximate) CCE. In particular, we show that the inapproximability of maximum clique precludes attaining \emph{any non-trivial sparsity} in polynomial time. Moreover, we strengthen our hardness results to apply in the low-precision regime as well via the \emph{planted clique conjecture}.
\end{abstract}


\pagenumbering{gobble}

\clearpage

\pagenumbering{arabic}


\section{Introduction}
\label{sec:intro}

One of the most influential results in the interface of algorithmic game theory and online learning is the realization that repeated play under \emph{no-regret}---a basic notion of hindsight rationality--- leads to a natural game-theoretic solution concept known as \emph{coarse correlated equilibrium (CCE)}~\citep{Hart00:Simple,Foster97:Calibrated}. Many ubiquitous algorithms guarantee the no-regret property, including (online) gradient descent and multiplicative weights update, and so one should expect CCE to arise from the repeated interaction of rational players---as it has been corroborated empirically~\citep{Nekipelov15:Econometrics,Kolumbus22:Auctions}. From an algorithmic standpoint, perhaps the most well-studied question that emerged from that connection concerns the number of iterations needed to approximate an equilibrium. Remarkably, although this problem traces back to the early pioneering works of~\citet{Blackwell56:analog} and ~\citet{Robinson51:iterative} in the 1950s, it remains poorly understood. This stands in contrast to the so-called adversarial regime, wherein a learner engages repeatedly with an adversarial environment acting so as to maximize the player's regret; in that setting, the minimax regret of the learner has long been resolved in the online learning literature~\citep{Littlestone94:Weighted}. Yet, it turns out that substantially improved guarantees are possible when the learner is instead competing against other learning players, as witnesses by a flurry of recent results (\emph{e.g.},~\citep{Daskalakis15:Near,Syrgkanis15:Fast,Rakhlin13:Optimization,Daskalakis21:Near,Piliouras22:Beyond}). Indeed, such learning dynamics have emerged as a key component in practical equilibrium computation~\citep{Brown19:Superhuman,Bowling15:Heads,Bakhtin22:Human,Perolat22:Mastering}, proving to be more scalable than traditional linear programming-based approaches.

In this paper, we study the iteration-complexity of no-regret learning in games under the constraint that the equilibrium reached approximates the optimal social welfare (or some other natural objective); henceforth, we will simply refer to such equilibria as near-optimal. Taking a step back, there has been tremendous interest in understanding the performance of no-regret learning in terms of welfare, primarily stemming from the \emph{price of anarchy} literature (\emph{e.g.},~\citep{Roughgarden15:Intrinsic,Blum08:Regret}), but here we ask an entirely different but equally fundamental question: 
\begin{quote}
\begin{center}
\emph{How many iterations are needed so that no-regret players converge to a near-optimal (approximate) equilibrium?}
\end{center}
\end{quote}
In terms of proving a lower bound, one natural approach revolves around the premise that players initially possess no information about the game, and in each iteration they receive only some limited utility feedback. We argue that there are certain caveats to such an approach. First, it does not apply to the usual centralized model of computation where the underlying game is given as part of the input. Furthermore, even in decentralized settings there is often a centralized party endeavoring to intervene and guide players to desirable outcomes~\citep{Mguni19:Coordinating,Kempe20:Inducing,Li20:End,Liu22:Inducing,Balcan13:Circumventing,Balcan14:Near}. Many models have been proposed that differ based on the amount of information gathered by the centralized party, as well as the way communication occurs between the different entities; each such model is arguably reasonable depending on the application. The question thus is how to come up with a lower bound that is not brittle to assumptions regarding the way information is distributed, and thereby applies to all such settings.

We address this by resorting to \emph{computational} lower bounds, thereby ruling out fast convergence to a near-optimal equilibrium even when the entire game is known in advance---subsuming the so-called \emph{full feedback} setting (recalled in~\Cref{sec:prels})---and players can fully coordinate; at first glance, it might seem counterintuitive that non-trivial lower bounds can be established under such permissive assumptions. In particular, we focus on perhaps the simplest class of games for which such questions become meaningful: two-player (general-sum) games represented in normal form; since we are aiming to prove lower bounds, concentrating on a simple class of games only makes the result stronger.

\subsection{Our results}

We establish tight computational lower bounds for the number of iterations needed for no-regret players to reach a near-optimal equilibrium in two-player games. To do so, a key observation that drives our approach is that no-regret learning produces, essentially by definition, a CCE with a particular structure: one expressed as a \emph{mixture} (that is, a convex combination) of product distributions. In particular, $T$ rounds of learning results in a mixture of $T$ product distributions. We call such a distribution \emph{$T$-sparse} (\Cref{def:sparse}). In this context, our main contribution is to prove hardness results for the problem of computing a near-optimal $T$-sparse CCE, which---by virtue of the observation above---immediately circumscribes the number of iterations for no-regret learning as well, even in the centralized model of computation. Even though this is a fundamental problem, to our knowledge, we are the first to examine its computational complexity as a function of $T$.


One special case of this problem is well-understood: a near-optimal $1$-sparse CCE is nothing other than a near-optimal \emph{Nash equilibrium}, treated in the seminal work of~\citet{Gilboa89:Nash} (and subsequently extended by~\citet{Conitzer08:New} and~\citet{Kothari18:Sum}), and shown to be $\NP$-complete. On the other end of the spectrum, assuming that each player has $n$ available actions, it is easy to see that any correlated distribution---and in particular any CCE---is $n$-sparse. In light of the well-known fact that the optimal CCE can be computed in polynomial time via a linear program (\Cref{prop:LP}), we see that there is a phase transition dictated by the sparsity parameter. In fact, our first main result shows that attaining non-trivial sparsity in polynomial time is impossible (subject to $\P \neq \NP$). Below, for an $n \times n$ two-player game $\cG$, we denote by $\sparsecce(\cG, T, \epsilon, \epswel)$ the problem of computing a $T$-sparse CCE with equilibrium gap at most $\epsilon$ (in an additive sense; see~\Cref{def:cce}) and welfare at least $\opt - \epswel$, where $\opt$ is the welfare attained by the optimal $T$-sparse CCE. (Further background is given later in~\Cref{sec:prels}.)

\begin{restatable}{theorem}{main}
    \label{theorem:main1}
    $\sparsecce(\cG, n^{1 - \epsilon}, n^{-c}, n^{-c})$ with respect to $n \times n$ games is $\NP$-hard for any constant $\epsilon > 0$ and some constant $c$.
\end{restatable}

This means that roughly $n$ iterations are needed for (computationally bounded) no-regret learners to converge to a CCE with $\poly(1/n)$ equilibrium and optimality gap; that is, the trivial upper bound of $n$ is essentially the best one can hope for. Further, a slightly stronger complexity assumption precludes even a sparsity of $n/2^{(\log n)^{1-\gamma}}$ for some constant $\gamma > 0$ (\Cref{cor:quasihardness}). It is worth noting that $T \defeq n$ iterations also represent a natural information-theoretic threshold: in the full feedback setting, to which our lower bounds readily apply, there is a trivial exploration protocol that enables each player to fully determine its own payoff matrix (by simply iterating over all rows or columns)---trivializing the problem at least in (two-player) zero-sum games.

In addition, \Cref{theorem:main1} establishes a complexity separation between $\sparsecce$ and $\unsparsecce$---the latter problem lifts the welfare constraint imposed by the former. Namely, since $\unsparsecce(\cG, T, 0)$ is in $\PPAD$ even for $T = 1$~\citep{Papadimitriou94:On}, $\sparsecce$ is harder (subject to $\coNP \neq \NP$~\citep{Johnson88:How}) for any sparsity $T \leq n^{1 - \epsilon}$.

Moreover, we strengthen~\Cref{theorem:main1} in two key aspects. First, we show that it applies under a broad class of objectives, beyond (utilitarian) welfare, which additionally includes the \emph{egalitarian} social welfare and each player's (individual) utility (\Cref{cor:egal-wel,cor:oneutil}). Second, \NP-hardness persists for any \emph{multiplicative} approximation to the objective (\Cref{cor:inapprox}). The key construction behind those results, \Cref{theorem:basic-emb}, implies similar hardness results for two other natural problems pertaining to sparse CCE: deciding uniqueness (\Cref{theorem:uniqueness}), and determining existence after excluding certain (joint) action profiles (\Cref{theorem:subset}); those two latter problems do not hinge on any underlying objective.

\paragraph{Technical approach} Compared to Nash equilibria, the crux in proving lower bounds for sparse CCE lies in introducing correlation between the players. Many natural reductions designed for Nash equilibria are of little use even for sparsity $T = 2$, which partly explains why the complexity of sparse CCE remains poorly understood. The key challenge is to identify a basic construction that handles near-optimal $T$-sparse CCE even when $T \gg 1$.

In this context, to prove~\Cref{theorem:main1}, we extend the reduction of~\citet{Gilboa89:Nash}  who proved $\NP$-hardness only when $T = 1$. In particular, they came up with a reduction from the decision version of the maximum clique problem ($\maxclique$) to $\sparsecce(\cG, 1, 0, 0)$. We establish a natural generalization of their reduction (\Cref{alg:maxclique}); namely, we show that computing a \emph{$2T$-approximation} to $\maxclique$ polynomially reduces to $\sparsecce(\cG, T, n^{-c}, n^{-c})$ (\Cref{theorem:algorithm1}). That is, the sparsity of the underlying CCE translates to a degradation in the resulting approximation factor. We are then able to rely on the celebrated inapproximability of $\maxclique$~\citep{Zuckerman07:Linear} (\Cref{theorem:clique-hardness}) to arrive at~\Cref{theorem:main1}. The overall reduction has various new technical aspects, discussed in~\Cref{sec:maxcliqueReduction}. The refinements to~\Cref{theorem:main1} described earlier are established by suitably adjusting this basic reduction (\Cref{sec:further-impl,sec:inapprox-obj}).

\paragraph{Low-precision regime} So far, we have focused on the regime where both the equilibrium and the optimality gap scale as $\poly(1/n)$---a common setting when it comes to equilibrium computation. No-regret learning is often employed in the so-called low-precision regime, which we identify with $\epseq, \epswel \geq 1/\polylog n$. In that setting, even Nash equilibria admit a quasipolynomial-time algorithm~\citep{Lipton03:Playing}, and so one cannot hope to prove---barring major complexity breakthroughs---\NP-hardness results. Instead, following an earlier work by~\citet{Hazan11:How}, we rely on the so-called \emph{planted clique conjecture} from average-case complexity (\Cref{conj:planted}; \Cref{sec:prels} provides a self-contained overview). We are then able to show the following quasipolynomial lower bounds.

\begin{theorem}
    \label{theorem:main2}
    Assuming that~\Cref{conj:planted} holds, the following problems require $n^{\Omega(\log n)}$ time with respect to $n \times n$ games:
    \begin{itemize}[noitemsep,leftmargin=0.5cm]
        \item $\sparsecce(\cG, T, (\log n)^{-c}, (\log n)^{-c})$ for any $T = \polylog n$ and some constant $c = c(T)$,
        \item $\sparsecce(\cG, T, c, c)$ for any $T = O(1)$ and some constant $c = c(T)$.
    \end{itemize}
\end{theorem}

The first lower bound is shown by relying on our previous construction behind~\Cref{theorem:main1}. The second one, which concerns the more permissive regime in which $\epsilon, \epswel = \Theta(1)$, adapts the reduction of~\citet{Hazan11:How} pertaining to optimal Nash equilibria. We provide the technical details in~\Cref{sec:low-precision}.

\subsection{Further related work}
\label{sec:related}

The notion of a sparse CCE---a mixture of product distributions (\Cref{def:sparse})---was recently studied by~\citet{Foster23:Hardness} in the context of Markov (aka. stochastic) games to rule out the existence of polynomial-time no-regret algorithms---with respect to potentially non-Markovian deviations (see also the work of~\citet{Peng24:Complexity}). This stands in contrast to games represented in normal form, where the existence of efficient no-regret algorithms has been long known tracing back to~\citet{Blackwell56:analog}. Yet, establishing non-trivial lower bounds for sparse CCE in normal-form games remains an open problem even for sparsity $T = 2$. It is worth highlighting that even though a CCE (without the sparsity constraint) can be computed exactly by solving a linear program~\citep{Papadimitriou08:Computing}, by far the most well-studied approach in the literature revolves around no-regret learning. This can be mostly attributed to the scalability, the minimal memory footprint, as well as the amenability to a distributed implementation of the latter approach, motivating the problem of sparse CCE. Besides this connection with no-regret learning, we argue that sparse CCE is a natural notion, worth examining in its own right, and ties to a long line of work on low-rank approximation in machine learning. 

A related notion of sparsity imposes instead a bound on the number of nonzero elements of the distribution---that is, the size of its support. Unlike~\Cref{def:sparse}, that latter notion is well-studied and understood (\emph{e.g.}, \citep{Babichenko14:Simple}). It is clear that a distribution $\vmu$ with $T$ nonzero entries is $T$-sparse per~\Cref{def:sparse}, but the opposite does not hold in general. From the viewpoint of no-regret learning, proving lower bounds pertaining to distributions with small support translates to the setting where each player selects a \emph{pure} strategy, while \Cref{def:sparse} accounts for mixed strategies as well. To further elaborate on this difference, it is known that $\Omega(\log n /\epsilon^2 )$ nonzero elements in the support are information-theoretically necessary for even the existence of an $\epsilon$-Nash equilibrium in zero-sum games~\citep{Feder07:Approximating}, in turn implying that $\Omega(\log n /\epsilon^2)$ iterations are needed when players select pure strategies. On the other hand, a consequence of the minimax theorem is that a $1$-sparse equilibrium always exists, and can also be computed efficiently in such games via linear programming; this means that no superpolynomial computational lower bounds for no-regret learning in zero-sum games can be shown in mixed strategies.

A long-standing challenge in general-sum games is \emph{equilibrium selection}: there could be a multiplicity of equilibria, and some are more reasonable than others. A common antidote---albeit certainly not the only one---is to identify an equilibrium maximizing the social welfare (or some other natural objective). This can be achieved in polynomial time even in multi-player (normal-form) games represented explicitly, which motivates investigating the complexity of $\sparsecce$---the focus of our work. However, this is no longer the case in succinct games, where maximizing welfare is typically \NP-hard~\citep{Papadimitriou08:Computing} (\emph{cf.}~\citet{Barman15:Finding})---let alone sparsity constraints. This is also the case for two-player \emph{extensive-form} games (\emph{e.g.},~\citep{Zhang22:Optimal}).

Another motivation for proving lower bounds revolving around $T$-sparse CCE is that while many techniques that accelerate equilibrium computation rely on no-regret learning dynamics, they do not strictly comply with the traditional online nature of the framework. A notable example is \emph{alternation}~\citep{Tammelin15:Solving,Wibisono22:Alternating,Cevher23:Alternation}, whereby players update their strategies sequentially---as opposed to simultaneous updates. Importantly, such techniques are captured through sparse CCE. Beyond computational considerations, it is worth pointing out an orthogonal line of work that has focused on query complexity aspects of (coarse) correlated equilibria (\emph{e.g.},~\citep{Goldberg16:Bounds,Babichenko15:Query,Maiti23:Query,Goldberg23:Lower}, and references therein).
\section{Preliminaries}
\label{sec:prels}

In this section, we introduce some necessary background on games and pertinent equilibrium concepts (\Cref{sec:games}), as well as the problem of approximating the maximum clique of a graph and the related planted clique problem (\Cref{sec:max-planted}).

\paragraph{Notation} We use boldface letters, such as $\vx$ and $\vy$, to denote vectors. Matrices are represented using boldface capital letters, such as $\mat{R}$ and $\mat{C}$. In particular, $\mat{I}_{n} \in \R^{n \times n}$ denotes the identity matrix, while $\mat{0}_{n} \in \R^{n \times n}$ denotes the all-zeroes matrix. We use subscripts to access the coordinates of a vector or the entries of a matrix. Superscripts (together with parantheses) are typically reserved for the discrete time index. For a vector $\vx$, we use $\|\vx\|$ to denote its (Euclidean) $\ell_2$ norm and $\|\vx\|_1$ for the $\ell_1$ norm. If $\vx \in \R^{n + m}$, for $n, m \in \N \defeq \{1, 2, \dots \}$, we let $\R^n \ni \vx_{i \leq n} \defeq (\vx_1, \dots, \vx_n)$ and $\R^m \ni \vx_{i \geq n+1} \defeq (\vx_{n+1}, \dots, \vx_{n+m})$. For $T \in \N$, we use the shorthand notation $[T] \defeq \{1, \dots, T\}$. For simplicity, we sometimes employ the notation $O(\cdot), \Omega(\cdot), \Theta(\cdot)$ to suppresses (absolute) constants.

\subsection{Two-player games and equilibrium concepts}
\label{sec:games}

We first provide the definition of a \emph{sparse} coarse correlated equilibrium (CCE) (\Cref{def:sparse,def:cce}). Below and throughout, we focus on two-player games, but those definitions can be readily generalized to multi-player games as well (\emph{e.g.},~\citep{Cesa-Bianchi06:Predictiona}).

\paragraph{Two-player games} By convention, we refer to the players as Player $x$ (for the ``row'' player) and Player $y$ (for the ``column'' player). We operate under the usual \emph{normal form} representation of a (two-player) game. Here, Player $x$ and Player $y$ have a finite (and nonempty) set of actions $\cA_x$ and $\cA_y$, respectively. The \emph{utility} of Player $x$ and Player $y$ under a pair of actions $(a_x, a_y) \in \cA_x \times \cA_y$ is given by $\mat{R}_{a_x, a_y}$ and $\mat{C}_{a_x, a_y}$, respectively, where $\mat{R}, \mat{C} \in \Q^{\cA_x \times \cA_y}$ are the payoff matrices of the game, which are given as part of the input. Players can randomize by specifying a \emph{mixed} strategy---a probability distribution over their set of actions. For a pair of mixed strategies $(\vx, \vy) \in \Delta(\cA_x) \times \Delta(\cA_y)$, the (expected) utility of Player $x$ and Player $y$ is given by $\E_{(a_x, a_y) \sim (\vx, \vy)} \mat{R}_{a_x, a_y} =  \langle \vx, \mat{R} \vy \rangle$ and $\E_{(a_x, a_y) \sim (\vx, \vy)} \mat{C}_{a_x, a_y} = \langle \vx, \mat{C} \vy \rangle$, respectively.

\paragraph{Sparse CCE} We are now ready to introduce the notion of a sparse CCE. In the sequel, we will denote by $\vx \otimes \vy = \vx \vy^\top$ the outer (tensor) product of $\vx$ and $\vy$.

\begin{definition}[Sparse (correlated) distributions]
    \label{def:sparse}
    Let $\vmu \in \Delta(\cA_x \times \cA_y)$ be a (correlated) distribution supported on $\cA_x \times \cA_y$. We say that $\vmu$ is \emph{$T$-sparse}, for $T \in \N$, if there exist $\vx^{(1)}, \dots, \vx^{(T)} \in \Delta(\cA_x)$, $\vy^{(1)}, \dots, \vy^{(T)} \in \Delta(\cA_y)$ and $(\alpha^{(1)}, \dots, \alpha^{(T)}) \in \Delta([T])$ such that $\vmu = \sum_{t=1}^T \alpha^{(t)} (\vx^{(t)} \otimes \vy^{(t)})$; that is, $\vmu$ is a \emph{mixture of $T$ product distributions}. If $\alpha^{(1)} = \alpha^{(2)} = \dots = \dots = \alpha^{(T)} = 1/T$, we say that $\vmu$ is \emph{uniform} $T$-sparse.
\end{definition}

It is worth noting here a related notion of sparsity which imposes a bound on the support of $\vmu$, discussed earlier in~\Cref{sec:related}. We next recall the definition of a coarse correlated equilibrium.

\begin{definition}[Coarse correlated equilibrium]
    \label{def:cce}
    A (correlated) distribution $\vmu \in \Delta(\cA_x \times \cA_y)$ is an \emph{$\epsilon$-coarse correlated equilibrium ($\epsilon$-CCE)} if for any deviation $(a_x', a_y') \in \cA_x \times \cA_y$,
    \begin{align}
        \E_{ (a_x, a_y) \sim \vmu} [\mat{R}_{a_x, a_y} ] &\geq \E_{(\cdot, a_y) \sim \vmu}[ \mat{R}_{a_x', a_y}] - \epsilon, \label{align:ccex} \\
        \E_{ (a_x, a_y) \sim \vmu} [\mat{C}_{a_x, a_y} ] &\geq \E_{(a_x, \cdot) \sim \vmu}[ \mat{C}_{a_x, a_y'}] - \epsilon.\label{align:ccey}
    \end{align}
\end{definition}
A $0$-CCE will be referred to as a CCE. Coarse correlated equilibria relax the notion of a \emph{correlated equilibrium}~\citep{Aumann74:Subjectivity}. In the latter, Player $x$ can deviate with respect to any possible function $\phi_x : \cA_x \to \cA_x $ (and similarly for Player $y$); that is, $a_x'$ and $a_y'$ in \eqref{align:ccex} and~\eqref{align:ccey} are to be replaced with $\phi_x(a_x)$ and $\phi_y(a_y)$, respectively. We further remark that a $1$-sparse CCE is, by definition, a \emph{Nash equilibrium} (in particular, when $\vmu$ is supported on a single element, we have a \emph{pure} Nash equilibrium).

\paragraph{Sparse equilibria and regret} There is an immediate but important connection between sparse CCE and the framework of no-regret learning. Namely, suppose that at every time $t \in [T]$ players select strategies $\vx^{(t)} \in \Delta(\cA_x) $ and $\vy^{(t)} \in \Delta(\cA_y)$. In the \emph{full feedback} setting, the players at time $t$ receive as feedback $\vu_x^{(t)} \defeq \mat{R} \vy^{(t)}$ and $\vu_y^{(t)} \defeq \mat{C}^\top \vx^{(t)}$, respectively.\footnote{In online learning, it is customary to assume that the utility feedback is revealed to the players upon selecting their strategies, but our results do not rest on such assumptions (see~\Cref{sec:related} for a further discussion).} The \emph{regret} of Player $x$ is defined as $\reg^T_x \defeq \max_{\vxstar \in \Delta(\cA_x)} \sum_{t=1}^T \langle \vxstar - \vx^{(t)}, \vu_x^{(t)} \rangle$, and similarly for Player $y$. The key connection now, which follows readily from the definitions, is that no-regret learning produces a uniform $T$-sparse $\max(\reg_x^T/T, \reg_y^T/T)$-CCE. As a concrete example, when both players employ multiplicative weights update, one obtains a uniform $O(\log ( |\cA| ) / \epsilon^2)$-sparse $\epsilon$-CCE, where $|\cA| \defeq \max(|\cA_x|, |\cA_y|)$.

Further, the non-uniform version of~\Cref{def:sparse} corresponds to a weighted notion of regret in which a weight $\alpha^{(t)}$ is attached to the $t$th summand, with $(\alpha^{(1)}, \dots, \alpha^{(T)}) \in \Delta([T])$. The latter has received ample of interest in prior work (\emph{e.g.},~\citep{Abernethy18:Faster,Brown19:Solving,Zhang22:Equilibrium}), not least because it typically performs better in practice. Of course, precluding efficient computation of sparse CCE under non-uniform mixtures can only be stronger.

\paragraph{Social welfare} The \emph{expected (social) welfare}, denoted by $\sw(\cdot)$, under a pair of mixed strategies $(\vx, \vy) \in \Delta(\cA_x) \times \Delta(\cA_y)$ is the sum of the players' utilities: $\sw(\vx \otimes \vy) \defeq \langle \vx, (\mat{R} + \mat{C}) \vy \rangle$. More generally, the welfare of a distribution $\vmu$ on $\cA_x \times \cA_y$ is $\sw(\vmu) \defeq \E_{(a_x, a_y) \sim \vmu} [ \mat{R}_{a_x, a_y} + \mat{C}_{a_x, a_y} ]$. While welfare is perhaps the most common objective considered in the literature, some of our lower bounds apply to other natural objectives as well (\Cref{cor:oneutil,cor:egal-wel}).

Without any essential loss, we will henceforth switch our attention to games in which $|\cA_x| = |\cA_y|$; this can always be enforced by introducing ``dummy'' actions. For notational convenience, we will also let $\cA_x, \cA_y \defeq \{ 1, 2, \dots, n \}$; that is, the action of each player is represented via an integer (the corresponding player will be made clear from the context). Such a game $\cG$ will oftentimes be referred to as $n \times n$. We now highlight the following fact.

\begin{proposition}
    \label{prop:LP}
    For any $n \times n$ (two-player) game, there is a polynomial-time algorithm that computes a CCE maximizing the welfare. Further, any distribution $\vmu \in \Delta([n] \times [n])$ is $n$-sparse.
\end{proposition}

Indeed, regarding the first part of the claim, it is well-known that computing a CCE maximizing the welfare can be expressed as a (polynomial) linear program---this holds more generally in normal-form games represented explicitly and under any linear objective. To justify the claim that any correlated distribution $\vmu$ is $n$-sparse, we define, for all $t \in [n]$,  $\alpha^{(t)} = \sum_{j=1}^n \vmu_{t, j}$; $\vx^{(t)}_i = 1$ for $i = t$ and $0$ otherwise; and $\vy^{(t)} \propto (\vmu_{t, 1}, \dots, \vmu_{t, n})$. By definition, $\sum_{t=1}^n \alpha^{(t)} (\vx^{(t)} \otimes \vy^{(t)}) = \vmu$, thereby implying that $\vmu$ is $n$-sparse. On the other hand, $\vmu$ cannot necessarily be expressed as a \emph{uniform} mixture of $n$ product distributions (\emph{e.g.}, when $\vmu$ is supported only on the diagonal).

\subsection{Maximum clique and planted clique}
\label{sec:max-planted}

Consider an undirected $n$-node graph $G = (V, E)$, where we let $V \defeq [n]$. A \emph{clique} in $G$ is a subset of the nodes $K \subseteq V$ such that for $\{i, j \} \in E$ for all $i, j \in K$ with $i \neq j$. The \emph{maximum clique} problem ($\maxclique$) takes as input an $n$-node graph, and asks for the maximum clique in $G$. The decision version of $\maxclique$ was featured in the original list of $\NP$-complete problems compiled by~\citet{Karp72:Reducibility}. Accordingly, commencing with a celebrated connection with probabilistic checkable proofs (PCPs), its (in)approximability received tremendous attention in theoretical computer science (\emph{e.g.},~\citep{Arora98:Proof,Bellare94:Improved,Bellare98:Free,Feige96:Interactive,Hastad99:Clique,Zuckerman07:Linear}), culminating in the hardness results of~\citet{Hastad99:Clique} and~\citet{Zuckerman07:Linear} stated below. We recall that if $\omega(G)$ is the size of the maximum clique in $G$, a (multiplicative) approximation of $\rho \geq 1$ is a clique of size at least $\omega(G)/\rho$. For our purposes, we will make use of the following hardness result~\citep[Theorem 1.1]{Zuckerman07:Linear}.

\begin{theorem}[\nakedcite{Zuckerman07:Linear}]
    \label{theorem:clique-hardness}
    It is \NP-hard to approximate $\maxclique$ with respect to an $n$-node graph to a factor of $n^{1 - \epsilon}$ for any constant $\epsilon > 0$.
\end{theorem}
The trivial approximation factor for $\maxclique$ is $n$ (since $\omega(G) \leq n$), and so~\Cref{theorem:clique-hardness} tells us that there is no algorithm significantly better than the trivial one. \citet{Zuckerman07:Linear} also showed, by derandomizing the earlier lower bound of~\citet{Khot01:Improved}, that it is $\NPtilde$-hard to approximate $\maxclique$ to a factor of $n/2^{(\log n)^{1-\gamma}}$ for some constant $\gamma > 0$~\citep[Theorem 1.3]{Zuckerman07:Linear}. $\NPtilde$ here is the quasipolynomial analogue of $\NP$, where quasipolynomial stands for $2^{\polylog n}$.

\paragraph{Planted clique problem} Related to $\maxclique$, the so-called \emph{planted} (aka. hidden) clique problem ($\plantedclique$) was first studied independently by~\citet{Jerrum92:Large} and~\citet{Kucera95:Expected}. Here, a random graph $G = G(n, \frac{1}{2}, k)$ is constructed as follows. First, an arbitrary subset of $k = k(n) \in [n]$ nodes is selected and a clique is placed on them. Then, every other edge is included independently with probability $\frac{1}{2}$, and the goal is to recover the planted clique. In a random Erd\H{o}s-R\'enyi graph the size of the largest clique is $(2 + o_n(1)) \log n$ with high probability (w.h.p.),\footnote{We say that an event $\mathcal{E}$ occurs with high probability if $\pr[\mathcal{E}] \geq 1 - 1/n^{\Theta(1)}$.} and so when $k = (2 + \epsilon) \log n$, finding the planted clique is information-theoretically possible by identifying the maximum clique of $G$. Now, for $k \geq (2 + \epsilon) \log n$, it is not hard to see that the planted clique can be identified in quasipolynomial time $n^{O(\log n)}$; the planted clique conjecture postulates that there is no significantly faster algorithm. In particular, we will use the following form of the conjecture.

\begin{conjecture}
    \label{conj:planted}
    Solving $\plantedclique$ w.h.p. when $k = \polylog n$ requires $n^{\Omega(\log n)}$ time.
\end{conjecture}
In the precise sense of~\Cref{lemma:smallk}, $\plantedclique$ becomes easier as $k$ grows; the current threshold for polynomial-time algorithms is $k = \Theta(\sqrt{n})$~\citep{Alon98:Finding}. It is believed---and there is strong evidence for it~\citep{Barak19:Nearly}---that no polynomial-time algorithm exists when $k \ll \sqrt{n}$; the weaker version of~\Cref{conj:planted} will suffice for our purposes. A number of prior works from diverse areas have relied on the planted clique conjecture to prove different hardness results, such as testing $k$-wise independence~\citep{Alon07:Testing}, finding clusters in social networks~\citep{Balcan13:Finding}, public-key encryption~\citep{Applebaum10:Public}, and sparse PCA~\citep{Berthet13:Complexity}. Closely related to our work is the result of~\citet{Hazan11:How} pertaining to the complexity of computing Nash equilibria with near-optimal welfare (see also~\citep{Minder09:Small,Austrin13:Inapproximability}), which we extend to $\sparsecce$ (\Cref{theorem:main2}).
\section{The complexity of computing near-optimal sparse CCE}
\label{sec:main}

In this section, we establish our main results concerning the complexity of $\sparsecce$ and other related problems in the $\poly(1/n)$ regime. The low-precision regime is treated in~\Cref{sec:low-precision}.


\subsection{The basic reduction from approximate \texorpdfstring{$\maxclique$}{maximum clique}}
 \label{sec:maxcliqueReduction}
We first establish a reduction from approximating $\maxclique$ to the problem of computing sparse (approximate) CCE with near-optimal welfare (\Cref{alg:maxclique}). In what follows, we recall that we refer to the latter computational problem as $\sparsecce(\cG, T, \epseq, \epswel)$, which takes as input a (two-player) game $\cG$ in normal form; a sparsity parameter $T \in \N$ per~\Cref{def:sparse}; the (CCE) equilibrium gap $\epseq$ per~\Cref{def:cce}; and the optimality gap $\epswel$, also in an additive sense (multiplicative (in)approximability is the subject of~\Cref{cor:inapprox}). Then, it returns (deterministically) a $T$-sparse $\epseq$-CCE with welfare at least $\opt - \epswel$, where $\opt$ is the welfare of the optimal $T$-sparse CCE. Before proceeding, it is worth noting that without the equilibrium constraint the problem becomes trivial: one can just return the maximum entry of $\mat{R} + \mat{C}$, which is of course $1$-sparse.

\paragraph{The reduction} Suppose that we are given as input an undirected graph $G = ([n], E)$ with \emph{adjacency matrix} $\mat{A} = \mat{A}(G) \in \R^{n \times n}$, defined such that
\[
    \mat{A}_{i, j} = 
    \begin{cases}
        1 & \text{if } \{i, j\} \in E \lor (i = j), \text{ and} \\
        0 & \text{otherwise}.
    \end{cases}
\]
(It is convenient to assume that each entry in the diagonal of $\mat{A}$ is also $1$.) The following reduction mostly relies on the earlier construction of~\citet{Gilboa89:Nash}, but with some new, distinct ingredients. In particular, our reduction is given as~\Cref{alg:maxclique}. \Cref{alg:maxclique} takes as input $G$, and constructs for each $k$ in~\Cref{line:for} a two-player game (\Cref{line:game}) defined by the $(2n) \times (2n)$ payoff matrices 
\begin{equation}
    \label{eq:game}
    \R^{(2n) \times (2n)} \ni \mat{R} \defeq 
    \frac{1}{2}  
    \begin{pmatrix}
       \mat{A} + \gamma \mat{I}_{n} & - k \mat{I}_{n} \\
        k \mat{I}_{n} & \mat{0}_{n \times n}
    \end{pmatrix}
    \text{ and }
    \R^{(2n) \times (2n)} \ni \mat{C} \defeq
    \frac{1}{2} 
    \begin{pmatrix}
       \mat{A} + \gamma \mat{I}_{n} & k \mat{I}_{n} \\
        - k \mat{I}_{n} & \mat{0}_{n \times n}
    \end{pmatrix},
\end{equation}
where we recall that $\mat{I}_{n} \in \R^{n \times n}$ denotes the identity matrix, $\mat{0}_{n \times n} \in \R^{n \times n}$ denotes the all-zeroes matrix, and parameter $\gamma \ll 1$ is specified in~\Cref{line:gamma}.\footnote{One can normalize the above payoff matrices so that all entries are in $[-1, 1]$, which is typically assumed when additive approximations are considered. All our hardness results apply as stated under that assumption.} (An example of $\cG = \cG(G)$ for the $4$-node graph $G$ of~\Cref{fig:graph} is given in~\Cref{fig:adj-matrix}.) Then, \Cref{line:sparsecce} invokes $\sparsecce(\cG, T, \epseq, \epswel)$ for a suitable choice of $\epseq$ and $\epswel$ (\Cref{line:eps}). It is worth pointing out that game~\eqref{eq:game} is \emph{symmetric}, in that $\mat{R} = \mat{C}^\top$, which will only make the lower bound stronger.

\begin{figure}[!ht]
    \centering
    \begin{minipage}{0.35\textwidth} 
    \begin{tikzpicture}
            \foreach \i in {1,...,4} {
                    \node[circle, draw, scale=1] (n\i) at (\i*90:1.5) {\i}; 
            }
            
            \foreach \i in {1,...,4} {
                    \foreach \j in {1,...,4} {
                            \ifnum\i=\j
                                    \draw[darkgray] (n\i) -- (n\j);
                            \else
                                    \ifnum\i=2 \ifnum\j=4 \else \draw (n\i) -- (n\j); \fi \else \ifnum\i=4 \ifnum\j=2 \else \draw (n\i) -- (n\j); \fi \else \draw (n\i) -- (n\j); \fi \fi
                            \fi
                    }
            }
    \end{tikzpicture}
    \caption{A $4$-node graph $G$.}
    \label{fig:graph}
    \end{minipage}
    \hspace{1cm} 
    \begin{minipage}{0.55\textwidth}
    \scriptsize 
    \begin{tabular}{ccccc|cccc}
     & \begin{tikzpicture}
         \node[draw, circle, scale=0.7] {1};
     \end{tikzpicture}  & \begin{tikzpicture}
         \node[draw, circle, scale=0.7] {2};
     \end{tikzpicture}  & \begin{tikzpicture}
         \node[draw, circle, scale=0.7] {3};
     \end{tikzpicture}  & \begin{tikzpicture}
         \node[draw, circle, scale=0.7] {4};
     \end{tikzpicture}  &   &   &   \\
       \begin{tikzpicture}
         \node[draw, circle, scale=0.7] {1};
     \end{tikzpicture} & \cellcolor{darkgray}$1 + \gamma$ & \cellcolor{darkgray}1 & 1 & 1 & $-k$ & 0 & 0 & 0 \\
      \begin{tikzpicture}
         \node[draw, circle, scale=0.7] {2};
     \end{tikzpicture}  & \cellcolor{darkgray}1 & \cellcolor{darkgray}$1 + \gamma$ & 1 & 0 & 0 & $-k$ & 0 & 0 \\
       \begin{tikzpicture}
         \node[draw, circle, scale=0.7] {3};
     \end{tikzpicture} & 1 & 1 & \cellcolor{darkgray}$1 + \gamma$ & \cellcolor{darkgray}1 & 0 & 0 & $-k$ & 0 \\
       \begin{tikzpicture}
         \node[draw, circle, scale=0.7] {4};
     \end{tikzpicture} & 1 & 0 & \cellcolor{darkgray}1 & \cellcolor{darkgray}$1 + \gamma$ & 0 & 0 & 0 & $-k$ \\
    \hhline{~|--------}
        & $k$ & 0 & 0 & 0 & 0 & 0 & 0 & 0 \\
        & 0 & $k$ & 0 & 0 & 0 & 0 & 0 & 0 \\
        & 0 & 0 & $k$ & 0 & 0 & 0 & 0 & 0 \\
        & 0 & 0 & 0 & $k$ & 0 & 0 & 0 & 0 \\
    \end{tabular}
    \caption{The payoff matrix $2 \cdot \mat{R} = 2 \cdot \mat{R}(G)$ per~\eqref{eq:game}.}
    \label{fig:adj-matrix}
    \end{minipage}
\end{figure}

\begin{algorithm}[!ht]
\caption{Reducing (approximate) $\maxclique$ to $\sparsecce$
}
\label{alg:maxclique}
\KwData{$n$-node graph $G$, sparsity $T \in \N$}
\KwResult{A $2T$-approximation to the maximum clique of $G$}
Set $K \leftarrow \{i \}$ for any $i \in [n]$ \\
\For{$k = 2T, 3T, \dots, \lfloor n/T \rfloor T $ \label{line:for} }{ 
    Set $\gamma < \frac{1}{40^2 k^6 (k+1)^2}$ \label{line:gamma} \\
    Construct the $(2n) \times (2n)$ (two-player) game $\cG = \cG(k, \gamma)$ per~\eqref{eq:game} \label{line:game}\\
    Set the optimality gap $\epswel \leftarrow \gamma^2 T/ 2$ and the equilibrium gap $\epsilon  \leftarrow k \gamma /2$ \label{line:eps}  \\
    $\sum_{t=1}^T \alpha^{(t)} ( \vx^{(t)} \otimes \vy^{(t)} ) \leftarrow \sparsecce(\cG, T, \epseq, \epswel)$ \label{line:sparsecce} \\
    Set $\hatvx_{i \leq n}^{(t)} = \nicefrac{\vx^{(t)}_{i \leq n}}{ \sum_{i=1}^n \vx_i^{(t)}}$ and $\hatvy_{j \leq n}^{(t)} = \nicefrac{\vy^{(t)}_{j \leq n}} {\sum_{j=1}^n \vy_j^{(t)}}$ for all $t \in [T]$ 
    \label{line:til} \\
    Determine $t^\star \in \arg \max_{1 \leq t \leq T} (\alpha^{(t)} )$ \label{line:tstar} \\
    Let $S \leftarrow \{ i \in [n] : \hatvx_i^{(t^\star)} \geq \frac{1}{\ell} - 40 \ell k \sqrt{\gamma} \}$, where $\ell \defeq \lfloor \alpha^{(t^\star)} k \rfloor$ \label{line:S} \\
    \If{ $S$ induces an $\ell$-clique in $G$}{
        Set $K \leftarrow S$
    }
}
\KwOut{$K \subseteq [n]$}
\end{algorithm}

It is clear that, with the possible exception of~\Cref{line:sparsecce}, all other steps of~\Cref{alg:maxclique} can be immediately implemented in time $\poly(n)$ ($T$ of course can be safely assumed to be polynomial in $n$). The main claim regarding~\Cref{alg:maxclique} is summarized below.

\begin{theorem}
    \label{theorem:algorithm1}
    For any $n$-node graph $G$ and sparsity parameter $T \in \N$, \Cref{alg:maxclique} returns a $2T$-approximation to the maximum clique of $G$.
\end{theorem}

The proof of this theorem is the subject of~\Cref{subsec:prooftheorem} below. Assuming~\Cref{theorem:algorithm1}, we arrive at the following polynomial-time reduction from approximate $\maxclique$ to $\sparsecce$.

\begin{corollary}
    There exists a polynomial-time reduction from $2T$-approximate $\maxclique$ to $\sparsecce(\cG, T, n^{-c}, n^{-c})$ for some constant $c$.
\end{corollary}

Thus, combining with the inapproximability of $\maxclique$ (\Cref{theorem:clique-hardness}), we obtain the following main implications, the first of which was stated earlier in the introduction.

\main*

\begin{corollary}
    \label{cor:quasihardness}
    $\sparsecce(\cG, n/2^{(\log n)^{1-\gamma}}, n^{-c}, n^{-c})$ with respect to $n \times n$ games is $\NPtilde$-hard for some constants $\gamma > 0$ and $c$.
\end{corollary}

To put those results into context, and as we pointed out earlier in~\Cref{prop:LP}, there is a polynomial-time algorithm for solving $\sparsecce(\cG,n,0,0)$ via linear programming.

\subsubsection{Proof of Theorem~\ref{theorem:algorithm1}}
\label{subsec:prooftheorem}

We split the main argument of the proof of~\Cref{theorem:algorithm1} into~\Cref{lemma:optimalCCE,lemma:bad-mass,lemma:small-prob,lemma:induced-clique} below. Let us first provide a basic high-level overview of the approach for the special case of $T = 1$, which goes back to~\citet{Gilboa89:Nash}. Perhaps the first attempt that comes to mind when trying to reduce from $\maxclique$ consists of using the adjacency matrix as the payoff matrix of each player. Of course, this reduction gives rise to uninteresting (pure) equilibria with optimal welfare: for an edge in $G$, say $\{1, 2\} \in E$, Player $x$ can select $1$ and Player $2$ can select $2$. This is where the component $k \mat{I}_n$ in~\eqref{eq:game} comes into play: it forces each player to select each action with at most a probability of $\approx 1/k$---for otherwise the other player could profitably deviate to that component. Further, since only the joint strategies corresponding to $\mat{A}$ have positive welfare, any Nash equilibrium with high welfare has to allocate most of its mass in that component. Anti-concentration, however, is still not enough, for the players may still end up supporting their strategies on a complete bipartite subgraph in $G$ (that is, each player on a different part). Incorporating the diagonal term $\gamma \mat{I}_n$ in the adjacency matrix turns out to address that issue as well, as it incentives players to select strategies that are close to each other. With that basic skeleton in mind, we next analyze the reduction under $\sparsecce$.

In what follows, we shall consider a fixed value of $k$, and show that if $G$ contains a clique of size $k$, then~\Cref{alg:maxclique} indeed outputs a clique of size $k/T$. We begin with a characterization of the set of optimal CCE of game~\eqref{eq:game}, $(\mat{R}, \mat{C})$, as a function of the sparsity parameter $T$. 

\begin{restatable}{lemma}{optimalCCE}
    \label{lemma:optimalCCE}
    Suppose that graph $G$ contains a clique of size $k \in \N$. Then, assuming that $k/T \in \N$, there is a (uniform) $T$-sparse CCE in $(\mat{R}, \mat{C})$ with welfare $1 + \gamma T/k$.
\end{restatable}

As a result, the welfare bound attained by $T$-sparse CCE according to~\Cref{lemma:optimalCCE} improves as $T$ increases, which will be crucial for the upcoming argument. Indeed, focusing on the $n \times n$ submatrix of~\eqref{eq:game} corresponding to the adjacency matrix $\mat{A}$, we observe that the diagonal entries have strictly larger utility than all other entries. The basic idea now is that one can use the multiple mixtures of the CCE in order to assign more probability mass in the diagonal elements than what is possible under a product distribution, without violating the equilibrium constraint. Namely, we partition the $k$-clique into $T$ (disjoint) subsets, and every mixture is set to be the uniform distribution over the corresponding subset. In turn, this is reflected in attaining higher welfare. An instance of this reasoning when $n = 4 = k$ and $T = 2$ is illustrated in~\Cref{fig:adj-matrix} (although that graph does not admit a $4$-clique). We now formalize this argument.

\begin{proof}[Proof of~\Cref{lemma:optimalCCE}]
    Let $K \subseteq V$ be a clique in $G$ of size $k$, which is assumed to exist. We construct a mixture of product distributions as follows. We partition $K$ into $T$ (disjoint) subsets $K^{(1)}, K^{(2)}, \dots, K^{(T)}$, each of size $k/T$. For every $t \in [T]$, we define $\vx^{(t)}$ and $\vy^{(t)}$ to be the uniform distributions over $K^{(t)}$. Then, we let $\vmu \defeq \frac{1}{T} \sum_{t=1}^T \vx^{(t)} \otimes \vy^{(t)}$.

    We first claim that $\vmu$ attains a welfare of $1 + \gamma T/k$. Indeed, by the above definition of $(\vx^{(t)}, \vy^{(t)})$, $\pr_{(i, j) \sim (\vx^{(t)}, \vy^{(t)}) }[i = j] = T/k$ for any $t \in [T]$. So, $\sw(\vx^{(t)} \otimes \vy^{(t)}) = 1 + \gamma T/k$, where we used the fact that $K^{(t)} \subseteq [n]$ forms a clique on $G$ (by construction). In turn, this implies that $\sw(\vmu) = \frac{1}{T} \sum_{t=1}^T \sw(\vx^{(t)} \otimes \vy^{(t)}) = 1 + \gamma T/k$.

    Next, we show that $\vmu$ is a CCE. By symmetry, it suffices to consider deviations of Player $y$. We first observe that the marginal of $\vmu$ with respect to Player $x$, $\frac{1}{T} \sum_{t=1}^T \vx^{(t)}$, is the uniform distribution over the clique $K$. Thus, any deviation of Player $y$ to an action $j \in [n]$ can yield an (expected) utility of at most $\frac{1}{2}(1 + \gamma/k)$; this is at most the utility obtained under $\vmu$, which we saw was $\frac{1}{2}(1 + \gamma T/k)$ (since the players receive the same utility when the game is restricted on $\mat{A}$). Further, any deviation to an action $n+1 \leq j \leq 2n$ can lead to a utility of at most $\frac{1}{2}$. This is again strictly worse than what obtained under $\vmu$ since $\gamma > 0$, completing the proof.
\end{proof}

It is interesting to note that, when $T \geq 2$, the distribution described above is \emph{not} a \emph{correlated equilibrium} (but merely a \emph{coarse} correlated equilibrium). Let us consider, for example, the extreme scenario where $T = k$, in which case $\vmu$ is supported solely on the diagonal elements of the clique $K$. Considering Player $y$, we now see that any deviation $j \mapsto n + j$ for $j \in K$ secures a utility of $\frac{k}{2}$, which is much higher than the utility obtained under $\vmu$. The difference here is that Player $y$ receives its own component of $\vmu$ before contemplating a deviation; but since $\vmu$ is supported only on the diagonal, Player $y$ knows the action of Player $x$ with probability $1$.\looseness-1

Next, to extract interesting information about $G$ from a CCE in~\eqref{eq:game}, we need to ensure that players are typically playing on the $n \times n$ submatrix corresponding to $\mat{A}$. To do so, the observation is that upon excluding that submatrix, the rest of the game is zero-sum, and hence the attained welfare is nullified. But we saw earlier in~\Cref{lemma:optimalCCE} that there is a CCE with high welfare, and so a near-optimal CCE cannot afford to assign too much mass outside of that $n \times n$ component, as reflected in the lemma below.\looseness-1

\begin{restatable}{lemma}{badmass}
    \label{lemma:bad-mass}
    Consider a $T$-sparse $\epseq$-CCE $\vmu = \sum_{t=1}^T \alpha^{(t)} (\vx^{(t)} \otimes \vy^{(t)})$ of $(\mat{R}, \mat{C})$ with $\sw(\vmu) \geq 1$. For each $t \in [T]$, we define $\hatvmu = \sum_{t=1}^T \alpha^{(t)} (\hatvx^{(t)} \otimes \hatvy^{(t)})$, where $\hatvx_{i \leq n}^{(t)} \defeq \nicefrac{\vx^{(t)}_{i \leq n}}{\sum_{i=1}^n \vx_i^{(t)}}$ and $\hatvy_{j \leq n}^{(t)} \defeq \nicefrac{ \vy^{(t)}_{j \leq n} }{\sum_{j=1}^n \vy_j^{(t)}}$. Then, it holds that $\hatvmu$ is an $(\epseq + 2 k \gamma)$-CCE of $(\mat{R}, \mat{C})$ with $\sw(\hatvmu) \geq \sw(\vmu)$. Further, by construction, $\hatvmu$ is supported only on $[n] \times [n]$.
\end{restatable}
We clarify that a distribution $\vmu \in \Delta([2n] \times [2n])$ is said to be supported only on $[n] \times [n]$ if $\vmu_{i, j} = 0$ when $i \notin [n]$ or $j \notin [n]$. We also note that the understanding above is that $\hatvx, \hatvy \in \Delta^{2n}$, although they are supported only on $[n]$. Further, if $\sum_{i=1}^n \vx_i^{(t)} = 0$, $\hatvx_{i \leq n}^{(t)}$ can be taken to be any point in $\Delta^{n}$; the same convention is followed with respect to $\hatvy_{j \leq n}^{(t)}$.

\begin{proof}[Proof of~\Cref{lemma:bad-mass}]
    We define $\delta^{(t)} \defeq 1 - \sum_{i=1}^n \vx_i^{(t)} \sum_{j=1}^n \vy_j^{(t)}$ for each $t \in [T]$. We will first show that $\sum_{t=1}^T \alpha^{(t)} \delta^{(t)} \leq \gamma$. Indeed, we can bound the welfare of each individual product distribution as
    \[
        \sw(\vx^{(t)} \otimes \vy^{(t)}) = \langle \vx^{(t)}, (\mat{R} + \mat{C} ) \vy^{(t)} \rangle = \sum_{i=1}^n \sum_{j=1}^n \vx_i^{(t)} (\mat{A}_{i, j} + \gamma \mathbbm{1} \{ i = j\} ) \vy_j^{(t)} \leq (1 + \gamma) (1 - \delta^{(t)}),
    \]
    by definition of $\mat{R}$ and $\mat{C}$ in~\eqref{eq:game}. Thus, we have
    \[
        1 \leq \sw(\vmu) = \sum_{t=1}^T \alpha^{(t)} \sw(\vx^{(t)} \otimes \vy^{(t)}) \leq (1 + \gamma) \left(1 - \sum_{t=1}^T \alpha^{(t)} \delta^{(t)} \right).
    \]
    So, $\sum_{t=1}^T \alpha^{(t)} \delta^{(t)} \leq \gamma$. We now define $\hatvx \in \Delta^{2n}$ such that $\hatvx_{i \leq n}^{(t)} = \vx^{(t)}_{i \leq n}/\sum_{i=1}^n \vx_i^{(t)}$ and $\hatvy \in \Delta^{2n}$ such that $\hatvy_{j \leq n}^{(t)} = \vy^{(t)}_{j \leq n}/\sum_{j=1}^n \vy_j^{(t)}$ for all $t \in [T]$. Accordingly, we define $\hatvmu \defeq \sum_{t=1}^T \alpha^{(t)} (\hatvx^{(t)} \otimes \hatvy^{(t)})$; by construction, each product of $\hatvmu$ is supported solely on $[n] \times [n]$. It is also clear that $\sw(\hatvmu) \geq \sw(\vmu)$ since, also by construction, $\hatvx_i^{(t)} \geq \vx_i^{(t)}$ and $\hatvy_j^{(t)} \geq \vy_j^{(t)}$ for all $t \in [T]$ and $i, j \in [n]$. As a result, it suffices to bound the CCE gap of $\hatvmu$. To that end, we will consider only deviations by Player $x$, as deviations by Player $y$ can be treated similarly. Since $\vmu$ is assumed to be an $\epsilon$-CCE, we have that for any $\vx \in \Delta^{2n}$,
    \begin{equation}
        \label{eq:CCE-def}
        \sum_{t=1}^T \alpha^{(t)} \langle \vx - \vx^{(t)}, \mat{R} \vy^{(t)}  \rangle \leq \epsilon.
    \end{equation}
    Now, we bound
    \begin{align}
        \langle \vx^{(t)}, \mat{R} \vy^{(t)} \rangle &\leq \sum_{i=1}^n \sum_{j=1}^n \vx_i^{(t)} \mat{R}_{i, j} \vy_j^{(t)} + \sum_{i=n+1}^{2n} \sum_{j=1}^n \vx_i^{(t)} \mat{R}_{i, j} \vy_j^{(t)} \notag \\
        &= \frac{1}{2} \sum_{i=1}^n \sum_{j=1}^n \vx_i^{(t)} (\mat{A}_{i, j} + \gamma \mathbbm{1} \{ i = j \} ) \vy_j^{(t)} + \frac{k}{2} \left( 1 - \sum_{i=1}^n \vx_i^{(t)} \right) \left(\sum_{j=1}^n \vy_j^{(t)} \right) \notag \\
        &\leq \frac{1}{2} \sum_{i=1}^n \sum_{j=1}^n \hatvx_i^{(t)} (\mat{A}_{i, j} + \gamma \mathbbm{1} \{ i = j \} ) \hatvy_j^{(t)} + \frac{k}{2} \left( 1 - \sum_{i=1}^n \sum_{j=1}^n \vx_i^{(t)} \vy_j^{(t)} \right) \notag \\
        &= \langle \hatvx^{(t)}, \mat{R} \hatvy^{(t)} \rangle + \frac{k}{2} \delta^{(t)}.\label{eq:ub-cce}
    \end{align}
Also, for any $\vx \in \Delta^{2n}$,
\begin{equation}
    \label{eq:lb-cce}
    \langle \vx, \mat{R} \vy^{(t)} \rangle = \langle \vx, \mat{R} \hatvy^{(t)} \rangle - \sum_{i=1}^{2n} \sum_{j=1}^{2n} \vx_i \mat{R}_{i, j} (\vy^{(t)}_j - \hatvy^{(t)}_j) \geq \langle \vx, \mat{R} \hatvy^{(t)} \rangle - k \delta^{(t)},
\end{equation}
where the inequality above uses the fact that $\| \vy^{(t)} - \hatvy^{(t)} \|_1 = 2 ( 1 - \sum_{j=1}^n \vy_j^{(t)})$. Indeed, when $\sum_{j=1}^n \vy_j^{(t)} = 0$, we have $\| \vy^{(t)} - \hatvy^{(t)} \|_1 = \|\vy^{(t)} \|_1 + \|\hatvy^{(t)} \|_1 = 2 = 2 ( 1 - \sum_{j=1}^n \vy_j^{(t)})$; in the contrary case, we have
\begin{align*}
    \|\vy^{(t)} - \hatvy^{(t)} \|_1 &= \|\vy^{(t)}_{j \geq n+1} \|_1 + \left\|\vy^{(t)}_{j \leq n} - \frac{\vy^{(t)}_{j \leq n}}{\sum_{j=1}^n \vy_j^{(t)}} \right\|_1 \\
    &= 1 - \sum_{j=1}^n \vy_j^{(t)} + \left( 1 - \frac{1}{\sum_{j=1}^n \vy_j^{(t)}} \right) \sum_{j=1}^n \vy_j^{(t)} = 2 \left( 1 - \sum_{j=1}^n \vy_j^{(t)} \right).
\end{align*}
Combining~\eqref{eq:ub-cce} and~\eqref{eq:lb-cce} with~\eqref{eq:CCE-def}, it follows that for any $\vx \in \Delta^{2n}$,
\[
    \sum_{t=1}^T \alpha^{(t)} \langle \vx - \hatvx^{(t)}, \mat{R} \hatvy^{(t)} \rangle \leq \epsilon + 2 k \sum_{t=1}^T \alpha^{(t)} \delta^{(t)} \leq \epsilon + 2k \gamma,
\]
and the same bound applies for $\sum_{t=1}^T \alpha^{(t)} \langle \vy - \hatvy^{(t)}, \mat{C}^\top \hatvx^{(t)} \rangle$. This shows that $\hatvmu$ is an $(\epsilon + 2k \gamma)$-CCE, as claimed. This completes the proof.
\end{proof}

\Cref{lemma:bad-mass} thus allows us to transform a near-optimal CCE into one solely supported on the $n \times n$ submatrix corresponding to $\mat{A}$, while incurring only a small degradation in the equilibrium quality---so long as $\gamma$ is sufficiently small. Assuming now a CCE with such a property, the next lemma shows that \emph{for each individual product distribution} comprising the underlying CCE, no action can be played with too high probability---at least when the corresponding weight $\alpha^{(t)}$ is sufficiently large. Otherwise, when $k$ is large, a deviating player could benefit by transitioning to the part of the game where that player's utility is $\frac{k}{2} \mat{I}_{n}$. This is made precise in the following lemma.

\begin{restatable}{lemma}{smallprob}
    \label{lemma:small-prob}
    Consider a $T$-sparse $\epsilon'$-CCE $\vmu \defeq \sum_{t=1}^T \alpha^{(t)} (\vx^{(t)} \otimes \vy^{(t)})$ of $(\mat{R}, \mat{C})$ such that each distribution $\vx^{(t)} \otimes \vy^{(t)}$ is supported only on $[n] \times [n]$. Then, it holds that $\vx^{(t)}_i, \vy_j^{(t)} \leq (1 + \gamma + 2 \epsilon')/(k \alpha^{(t)})$ for all $t \in [T]$ and $i, j \in [n]$.
\end{restatable}
An important aspect of the lemma above is that one can obtain meaningful information about each product distribution comprising $\vmu$ using only properties of $\vmu$. Assuming for simplicity that $\alpha^{(1)} = \dots = \alpha^{(T)}$, we see that the upper bound on the probabilities in~\Cref{lemma:small-prob} grows roughly as $T/k$---that is, it quickly deteriorates when $T \gg 1$. 

\begin{proof}[Proof of~\Cref{lemma:small-prob}]
    Fix some $i \in [n]$. We see that a deviation from Player $y$ to the $n + i$ column would give that player a utility of $\mat{C}_{i, n+i} \left( \sum_{t=1}^T \alpha^{(t)} \vx^{(t)}_i \right)$. But the current utility of Player $y$ under $\vmu$ is at most $\frac{1}{2} (1 + \gamma) $ since $\vmu$ is assumed to be supported only on $[n] \times [n]$. As a result, since $\vmu$ is an $\epsilon'$-CCE, this implies that $k \sum_{t=1}^T \alpha^{(t)} \vx_i^{(t)} \leq (1 + \gamma) + 2 \epsilon'$, which in turn yields
    \[
         \max_{t \leq T} \alpha^{(t)} \vx_i^{(t)} \leq \frac{1 + \gamma + 2 \epsilon'}{k}.
    \]
    Similar reasoning gives the same bound on $\max_{ t \leq T} \alpha^{(t)} \vy_j^{(t)}$ for any $j \in [n]$.
\end{proof}

Finally, the next lemma shows how one can extract a clique on $G$ using a product distribution with a suitable lower bound on its welfare (which follows from~\Cref{lemma:optimalCCE}) and an upper bound on the probability of its action (as in~\Cref{lemma:small-prob}).

\begin{restatable}{lemma}{inducedclique}
    \label{lemma:induced-clique}
    Consider a product distribution $\vx \otimes \vy$ supported only on $[n] \times [n]$ with welfare at least $1 + \gamma/\ell - 6 \gamma^2 T^2$ and $\vx_i, \vy_j \leq (1 + 6 k \gamma)/\ell$ for any $i, j \in [n]$, where $\ell \defeq \alpha^{(t)} k$ with $\alpha^{(t)} \geq \frac{1}{T}$. If $\gamma < \frac{1}{40^2 \ell^4 (\ell +1 )^2 k^2}$ and $k \geq 2 T$, the $\ell$ largest coordinates of $\vx$ induce an $\ell$-clique on $G$.
\end{restatable}

Before we proceed with the proof, we will make use of a simple claim stated and proven below.

\begin{claim}
    \label{claim:squared-bound}
    Let $\vx \in \R^{n}$ be such that $0 \leq \vx_i \leq 1/\ell$ for all $i \in [n]$ and $\sum_{i=1}^n \vx_i = p \in (0, n/\ell]$. Then, $\sum_{i=1}^n \vx_i^2 \leq p/\ell$.
\end{claim}

\begin{proof}
    Let $\vx \in \R^n$ be a maximizer of $\sum_{i=1}^n \vx_i^2$ subject to the constraints specified above. Let $S \subseteq [n]$ be the nonzero indexes of $\vx$. For the sake of contradiction, suppose that there exist $i, i' \in S$ with $i \neq i'$ such that $\vx_{i} > \vx_{i'}$ and $\vx_i < 1/\ell$. Then, there is a sufficiently small $\epsilon > 0$ such that $\vx_{i} + \epsilon \leq 1/\ell$ and $\vx_{i'} - \epsilon \geq 0$, and at the same time $(\vx_i + \epsilon)^2 + (\vx_{i'} - \epsilon)^2 > \vx_i^2 + \vx_{i'}^2$. This contradicts the optimality of $\vx$. We conclude that there are two possible cases. First, if $\vx_i = p/|S|$ for all $i \in S$, we have $\sum_{i \in S} \vx_i^2 = p^2/|S| \leq p/\ell$ (since $p/|S| \leq 1/\ell$ by feasibility). Otherwise, there is a coordinate $i \in S$ such that $\vx_i < 1/\ell$ while $\vx_{i'} = 1/\ell$ for $i' \in S \setminus \{i\}$. Thus, $\sum_{i \in S} \vx_i^2 \leq (|S| - 1)/\ell^2 + ( p - (|S| - 1)/\ell)^2 < (|S| - 1)/\ell^2 + ( p - (|S| - 1)/\ell) / \ell = p/\ell$.
\end{proof}

\begin{proof}[Proof of~\Cref{lemma:induced-clique}]
    Given that $\vx \otimes \vy$ is supported only on $[n] \times [n]$, we have
    \[
        (1 + \gamma) \sum_{i=1}^n \vx_i \vy_i + \left(1 - \sum_{i=1}^n \vx_i \vy_i \right) \geq \sw(\vx \otimes \vy) \geq 1 + \gamma / \ell - 6 \gamma^2 T^2,
    \]
    which implies that $\sum_{i=1}^n \vx_i \vy_i \geq 1/\ell - 6 \gamma T^2$. Further, \Cref{claim:squared-bound} implies that $\sum_{i=1}^n \vx_i^2 \leq (1 + 6 k \gamma) /\ell $ and $\sum_{i=1}^n \vy_i^2 \leq (1 + 6 k \gamma) /\ell$. Thus, combining the previous inequalities,
    \begin{equation}
        \label{eq:dist-xy}
        \|\vx - \vy \|^2 = \sum_{i=1}^n \vx_i^2 + \sum_{i=1}^n \vy_i^2 - 2 \sum_{i=1}^n \vx_i \vy_i \leq \frac{12 k \gamma + 12 \gamma T^2 }{\ell} \leq \frac{15 k^2 \gamma}{\ell} .
    \end{equation}
    Moreover,
    \begin{align*}
        \frac{1}{\ell} \leq \langle \vx, \vy \rangle = \langle \vx, \vy - \vx \rangle + \|\vx\|^2 &\leq \|\vx\| \|\vx - \vy\| + \|\vx\|^2 \\
        &\leq \sqrt{ \frac{1 + 6 k \gamma}{\ell}} \sqrt{\frac{15 k^2 \gamma}{\ell}} + \|\vx\|^2 \\
        &\leq \frac{4 k \sqrt{\gamma}}{\ell} + \|\vx\|^2
    \end{align*}
    since $\sqrt{1 + 4 k \gamma} \leq \sqrt{16/15}$. Rearranging,
    \begin{equation}
        \label{eq:lb-x2norm}
        \|\vx\|^2 \geq \frac{1}{\ell} \left( 1 - 4 k \sqrt{\gamma}\right).
    \end{equation}
Now, let $S \defeq \{i \in [n] : \vx_i \geq (1 - \gamma')/\ell \}$ for $\gamma' \defeq 40 \ell^2 k \sqrt{\gamma}$, and $p \defeq \sum_{i \in S} \vx_i$. Applying~\Cref{claim:squared-bound} to both $S$ and $[n] \setminus S$,
\[
     \|\vx\|^2 = \sum_{i \in S} \vx_i^2 + \sum_{ i \in [n] \setminus S} \vx_i^2 
     \leq p \frac{1 + 6 k \gamma}{\ell} + (1 -p) \frac{1 - \gamma'}{\ell}.
\]
Combining with~\eqref{eq:lb-x2norm}, we get
\[
     p \geq \frac{\gamma' - 4 k \sqrt{\gamma} }{6 k \gamma + \gamma'} \geq 1 - \frac{1}{4 \ell^2},
\]
where we used that $\gamma' = 40 \ell^2 k \sqrt{\gamma}$. We now claim that $|S| = \ell$. Indeed, since $(1 - \gamma')/\ell \leq \vx_i \leq (1 + 6 k \gamma) /\ell$ for all $i \in S$, we have
\[
    \frac{1 - \gamma'}{\ell} |S| \leq \sum_{i \in S} \vx_i \leq 1 \implies |S| \leq \frac{\ell}{1 - \gamma'} = \frac{\ell}{ 1 - 40 \ell^2 k \sqrt{\gamma} } < \ell + 1
\]
for $\gamma < \frac{1}{40^2 \ell^4 (\ell+1)^2 k^2}$, and (since $p = \sum_{i \in S} \vx_i \geq 1 - 1/(4 \ell^2)$)
\[
    1 - \frac{1}{4 \ell^2} \leq \sum_{i \in S} \vx_i \leq |S| \frac{1 + 6 k \gamma}{ \ell} \implies |S| > \ell - 1.
\]
Thus, $|S| = \ell$. Finally, we show that $S$ is an $\ell$-clique on $G$. We first note that, by~\eqref{eq:dist-xy},
\[
  \vy_j \geq \vx_j - \sqrt{ \frac{15 k^2 \gamma}{\ell}} \geq \frac{1 - \gamma'}{\ell} - 4 k \sqrt{\frac{\gamma}{\ell}}  \geq \frac{1}{\ell} - \frac{1.5}{\ell(\ell+1)} \geq \frac{1}{2\ell}
\] 
 for all $j \in S$, where the penultimate inequality follows from our choice of $\gamma$ and $\gamma'$, and the last inequality uses that $\ell \geq 2$. If we assume for the sake of contradiction that there exist $i', j' \in S$ such that $\mat{A}_{i', j'} = 0$, that is, $S$ is not a clique on $G$, then
 \begin{align*}
     \sw(\vx \otimes \vy) &= \sum_{i=1}^n \sum_{j = 1}^n \vx_i (\mat{A}_{i, j} + \gamma \mathbbm{1} \{ i = j \} ) \vy_j \\ &\leq (1 + \gamma) \sum_{i=1}^n \vx_i \vy_i + \sum_{i=1}^n \sum_{j \neq i} \mathbbm{1} \{\{i, j\} \neq \{i', j'\}\} \vx_i \mat{A}_{i, j} \vy_j + 2 
\vx_{i'} \mat{A}_{i', j'} \vy_{j'} \\ &\leq (1 + \gamma) \frac{1 + 6 k \gamma}{\ell} + 1 - \frac{1 + 6 k\gamma}{\ell} - \frac{1}{2 \ell^2} \\
     &\leq 1 + \frac{\gamma}{\ell} - \frac{1}{\ell^2} < 1 + \frac{\gamma}{\ell} - 6 \gamma^2 T^2,
 \end{align*}
 where the last inequality follows from our choice of $\gamma$. This is a contradiction since $\sw(\vx \otimes \vy) \geq 1 + \gamma/\ell - 6 \gamma^2 T^2$.
 \end{proof}

We are now ready to combine~\Cref{lemma:optimalCCE,lemma:bad-mass,lemma:small-prob,lemma:induced-clique} in order to prove~\Cref{theorem:algorithm1}.

\begin{proof}[Proof of~\Cref{theorem:algorithm1}]
    It suffices to show that if $G$ contains a $k$-clique, then~\Cref{line:S} outputs a $k/T$-clique. Indeed, if $\omega(G) \leq 2 T$, then~\Cref{alg:maxclique} trivially returns a $2T$-approximation; otherwise, it suffices to consider the largest iteration of the algorithm such that $k \leq \omega(G)$. Consider now for that iteration the $T$-sparse $\epsilon$-CCE, $\vmu = \sum_{t=1}^T \alpha^{(t)} (\vx^{(t)} \otimes \vy^{(t)})$, given as output by $\sparsecce$ in~\Cref{line:sparsecce}; we can assume here that $\alpha^{(t)} > 0$ for all $t \in [T]$. Let $\hatvmu \defeq \sum_{t=1}^T \alpha^{(t)} (\hatvx^{(t)} \otimes \hatvy^{(t)})$, where $(\hatvx^{(t)})_{t \leq T}$ and $(\hatvy^{(t)})_{
t \leq T}$ are as in~\Cref{line:til}. By~\Cref{lemma:optimalCCE}, we know that there is a $T$-sparse CCE with welfare at least $1 + \gamma T/k$, and so since the optimality gap $\epswel$ is small enough (by \Cref{line:eps}), we get that $\sw(\vmu) \geq 1$. Thus, we can apply~\Cref{lemma:bad-mass} to conclude that $\hatvmu$ is an $(\epsilon + 2k \gamma)$-CCE, which is a $(5 k \gamma /2)$-CCE by our choice of $\epsilon$ (\Cref{line:eps}). Next, \Cref{lemma:small-prob} implies that $\hatvx_i^{(t)} \leq (1 + \gamma + 5 k \gamma)/(\alpha^{(t)} k)$ for all $i \in [n]$ and $t \in [T]$. Using~\Cref{claim:squared-bound} and Cauchy-Schwarz, we have
\[
    \sum_{i=1}^n \hatvx_i^{(t)} \hatvy_i^{(t)} \leq \sqrt{ \left( \sum_{i=1}^n (\hatvx_i^{(t)})^2 \right) \left( \sum_{i=1}^n (\hatvy_i^{(t)})^2 \right) } \leq \frac{1 + \gamma + 5 k \gamma}{\alpha^{(t)} k }.
\]
Thus, for all $t \in [T]$,
\begin{equation}
    \label{eq:sw-bound}
    \sw(\hatvx^{(t)} \otimes \hatvy^{(t)}) \leq 1 + \frac{\gamma + \gamma^2 + 5 k \gamma^2}{\alpha^{(t)} k}.
\end{equation}
By~\Cref{lemma:optimalCCE,lemma:bad-mass}, we have
\begin{align}
    \alpha^{(t)} \sw(\hatvx^{(t)} \otimes \hatvy^{(t)}) &\geq 1 + \frac{\gamma T}{k} - \epswel - \sum_{\tau \neq t} \alpha^{(\tau)} \sw(\hatvx^{(\tau)} \otimes \hatvy^{(\tau)}) \notag \\
    &\geq \alpha^{(t)} + \frac{\gamma T}{k} - \epswel - (T-1) \frac{\gamma + \gamma^2 + 5 k \gamma^2}{k} \label{align:sw-bound} \\
    &\geq \alpha^{(t)} + \frac{\gamma}{k} - 5 \gamma^2 T - \frac{T}{k} \gamma^2 - \epswel \geq \alpha^{(t)} + \frac{\gamma}{k} - 6 \gamma^2 T, \label{align:epswel}
\end{align}
where~\eqref{align:sw-bound} uses~\eqref{eq:sw-bound} and the fact that $\alpha^{(t)} = 1 - \sum_{\tau \neq t} \alpha^{(\tau)} $, and~\eqref{align:epswel} follows from our choice of $\epswel$ and the fact that $k \geq 2$. Now, let $t^\star \in [T]$ be such that $\alpha^{(t^\star)} \geq 1/T$; \Cref{line:tstar} returns such a $t^\star$. Then, \eqref{align:epswel} implies that $\sw(\hatvx^{(t^\star)} \otimes \hatvy^{(t^\star)}) \geq 1 + \frac{\gamma}{\alpha^{(t^\star)} k} - 6 \gamma^2 T^2$ since $\alpha^{(t^\star)} \geq 1/T$. Finally, we can use \Cref{lemma:induced-clique} with respect to $\hatvx^{(t^\star)} \otimes \hatvy^{(t^\star)}$ to obtain an $(\alpha^{(t^\star)} k)$-clique, and thereby a $k/T$-clique in $G$ since $\alpha^{(t^\star)} \geq 1/T$. This completes the proof.
\end{proof}


\subsection{Further implications}
\label{sec:further-impl}

Building on our previous reduction, we first obtain some further lower bounds for natural problems related to sparse CCE, extending some corresponding results of~\citet{Gilboa89:Nash}. The reduction here serves as a warm-up for the upcoming one in~\Cref{sec:inapprox-obj}, which will enable us to capture other natural objectives, beyond welfare, and preclude any multiplicative approximation.

In this context, the problem $\uniquesparsecce(\cG, T)$ asks whether the two-player game $\cG$ given as input admits a unique \emph{uniform} $T$-sparse CCE (recall~\Cref{def:sparse}); no approximation error is allowed here, for otherwise there are trivially multiple sparse CCE. The reason why we need to posit a uniform $T$-sparse distribution is explained more after the proof of~\Cref{theorem:uniqueness}. We further point out that $\uniquesparsecce$ does not hinge on an underlying objective. As before, it is easy to see that without the sparsity constraint there is a polynomial-time algorithm based on linear programming. We will show that obtaining strongly sublinear sparsity is again computationally hard (\Cref{theorem:uniqueness}).

The basic idea here is to augment the game given in~\eqref{eq:game} with an additional action for each player, say $\newaction$, such that the following property holds: if the resulting game has a unique uniform $T$-sparse CCE, namely $(\newaction, \newaction)$, then $G$ does not contain a clique of size $k$; otherwise, $G$ must contain a clique of size $\Omega(k/T)$. Relying then on the hardness result of~\citet{Zuckerman07:Linear}, we arrive at the following theorem.

\begin{figure}[!ht]
\tiny
\centering
\begin{tabular}{ccccc|cccc|c}
     & \begin{tikzpicture}
         \node[draw, circle, scale=0.7] {1};
     \end{tikzpicture}  & \begin{tikzpicture}
         \node[draw, circle, scale=0.7] {2};
     \end{tikzpicture}  & \begin{tikzpicture}
         \node[draw, circle, scale=0.7] {3};
     \end{tikzpicture}  & \begin{tikzpicture}
         \node[draw, circle, scale=0.7] {4};
     \end{tikzpicture}  &   &   &  & & $\newaction$  \\
       \begin{tikzpicture}
         \node[draw, circle, scale=0.7] {1};
     \end{tikzpicture} & \cellcolor{darkgray}$(1 + \gamma, 1 + \gamma)$ & \cellcolor{darkgray}(1, 1) & (1, 1) & (1, 1) & $(-k, k)$ & (0, 0) & (0, 0) & (0, 0) & $(-2r, 2r)$ \\
      \begin{tikzpicture}
         \node[draw, circle, scale=0.7] {2};
     \end{tikzpicture}  & \cellcolor{darkgray}(1, 1) & \cellcolor{darkgray}$(1 + \gamma, 1 + \gamma)$ & (1, 1) & (0, 0) & 0 & $(-k, k)$ & (0, 0) & (0, 0) & $(-2r, 2r)$ \\
       \begin{tikzpicture}
         \node[draw, circle, scale=0.7] {3};
     \end{tikzpicture} & (1, 1) & (1, 1) & \cellcolor{darkgray}$(1 + \gamma, 1 + \gamma)$ & \cellcolor{darkgray}(1,1) & (0, 0) & (0, 0) & $(-k, k)$ & (0, 0) & $(-2r, 2r)$ \\
       \begin{tikzpicture}
         \node[draw, circle, scale=0.7] {4};
     \end{tikzpicture} & (1, 1) & (0, 0) & \cellcolor{darkgray}(1, 1) & \cellcolor{darkgray}$(1 + \gamma, 1 + \gamma)$ & (0, 0) & (0, 0) & (0, 0) & $(-k, k)$ & $(-2r, 2r)$ \\
    \hhline{~|--------}
        & $(k, -k)$ & (0, 0) & (0, 0) & (0, 0) & (0, 0) & (0, 0) & (0, 0) & (0, 0) & $(-2r, 2r)$ \\
        & (0, 0) & $(k, -k)$ & (0, 0) & (0, 0) & (0, 0) & (0, 0) & (0, 0) & (0, 0) & $(-2r, 2r)$ \\
        & (0, 0) & (0, 0) & $(k, -k)$ & (0, 0) & (0, 0) & (0, 0) & (0, 0) & (0, 0) & $(-2r, 2r)$ \\
        & (0, 0) & (0, 0) & (0, 0) & $(k, -k)$ & (0, 0) & (0, 0) & (0, 0) & (0, 0) & $(-2r, 2r)$ \\
    \hhline{~|---------}
    $\newaction$ & $(2r, -2r)$ & $(2r, -2r)$ & $(2r, -2r)$ & $(2r, -2r)$ & $(2r, -2r)$ & $(2r, -2r)$ & $(2r, -2r)$ & $(2r, -2r)$ & \cellcolor{darkgray}$(2r, 2r)$ 
    \end{tabular}
    \caption{Payoff matrices of $\cG'$ (after multiplying each entry by $2$) based on graph $G$ of~\Cref{fig:graph}.}
    \label{fig:new-matrix}
\end{figure}

\begin{restatable}{theorem}{uniqueness}
    \label{theorem:uniqueness}
    $\uniquesparsecce(\cG, n^{1 - \epsilon})$ is $\NP$-hard with respect to $n \times n$ games for any constant $\epsilon > 0$.
\end{restatable}

\begin{proof}
    
We will use the fact that distinguishing whether an $n$-node graph $G$ has a clique of size at most $n^{\epsilon}$ or at least $n^{1 - \epsilon}$ is \NP-hard for any constant $\epsilon > 0$~\citep{Zuckerman07:Linear}. We proceed as follows. We augment the game given in~\eqref{eq:game} with an additional action, $\newaction$, for each player. We will refer to the resulting game as $\cG'$. When both players select $\newaction$, they both receive a utility of $r \defeq (1 + \gamma T/k)/2$. When only one player selects $\newaction$, that player receives $r$ while the other player receives $-r$ (see~\Cref{fig:new-matrix}). It is clear that $(\newaction, \newaction)$ is a (pure) Nash equilibrium, and hence a $T$-sparse CCE. Moreover, \Cref{lemma:optimalCCE} shows that when $G$ contains a $k$-clique, there is a uniform $T$-sparse correlated distribution $\vmu$ in $\cG$ supported only on the $n \times n$ submatrix corresponding to $\mat{A}$, in which each player obtains a utility of $r$. It is easy to see that $\vmu$ induces a (uniform) $T$-sparse CCE in $\cG'$ as well. What remains to show thus is that when $G$ does not contain an $\Omega(k/T)$ clique, $(\newaction, \newaction)$ is the only uniform $T$-sparse CCE of $\cG'$.

Indeed, since a player can always switch to playing $\newaction$, which secures that player a utility of $r$, any $T$-sparse CCE in $\cG'$ must have welfare at least $2r = 1 + \gamma T / k$. Following the argument of~\Cref{subsec:prooftheorem}, when $\gamma$ is sufficiently small, we will show that this can only happen under $(\newaction, \newaction)$ or when $G$ contains an $\Omega(k/T)$-clique. In proof, let $\vmu = \frac{1}{T} \sum_{t=1}^T \vx^{(t)} \otimes \vy^{(t)}$ be a CCE in $\cG'$ distinct from $(\newaction, \newaction)$, which we argued must satisfy $\sw(\vmu) \geq 1 + \gamma T / k$. We consider now the product $t \in [T]$ that maximizes $\sw(\vx^{(t)} \otimes \vy^{(t)})$ subject to the constraint that the probability mass assigned to the $n \times n$ submatrix corresponding to $\mat{A}$ is (strictly) positive. It is clear that such $t$ exists, and $\sw(\vx^{(t)} \otimes \vy^{(t)}) \geq 1 + \gamma T/k$. We will now argue that $p \defeq \pr_{i \sim \vx^{(t)}, j \sim \vy^{(t)}} [i \leq n, j \leq n ] \geq 1 - 2\gamma$. Let $p_{\newaction} = \pr_{i \sim \vx^{(t)}, j \sim \vy^{(t)}} [i = \newaction, j = \newaction]$. Then, it follows that $\sw(\vx^{(t)} \otimes \vy^{(t)}) \leq p (1 + \gamma) + p_{\newaction} (1 + \gamma T /k)$ since $\cG'$ is, by construction, zero-sum when excluding $(\newaction, \newaction)$ and the $n \times n$ submatrix corresponding to $\mat{A}$. Thus, $p + p_{\newaction} \geq 1 - \gamma$. Further, a union bound implies that $p + p_{\newaction} + 2 p p_{\newaction} \leq 1$. Indeed, we have
\begin{align*}
    \pr_{i \sim \vx^{(t)}, j \sim \vy^{(t)}} [i = \newaction, j \neq \newaction] &= \pr_{i \sim \vx^{(t)}} [i = \newaction] \pr_{j \sim \vy^{(t)}} [j \neq \newaction] \\
    &\geq \pr_{i \sim \vx^{(t)}, j \sim \vy^{(t)}} [i = \newaction, j = \newaction] \pr_{i \sim \vx^{(t)}, j \sim \vy^{(t)}} [i \leq n, j \neq \newaction]  \\
    &\geq p_{\newaction} p,
\end{align*}
and the same lower bound applies to $\pr_{i \sim \vx^{(t)}, j \sim \vy^{(t)}} [i \neq \newaction, j = \newaction]$. Combining with the fact that $\sw(\vx^{(t)} \otimes \vy^{(t)}) \geq 1 + \gamma T/k$, we conclude that
\[
    1 + \gamma \frac{T}{k} \leq 1 - 2 p p_{\newaction} + p \gamma + p_{\newaction} \gamma \frac{T}{k} \leq 1 + \gamma \frac{T}{k} - 2 p p_{\newaction} + p \gamma \implies p_{\newaction} \leq \frac{\gamma}{2},
\]
where we used that $p > 0$ (by assumption of how $t \in [T]$ was chosen). Next, we define $\hatvx_{i \leq n}^{(t)} =  \nicefrac{\vx^{(t)}_{i \leq n}}{\sum_{i=1}^n \vx_i^{(t)}}$ and $\hatvy_{j \leq n}^{(t)} =  \nicefrac{\vy^{(t)}_{j \leq n}}{\sum_{j=1}^n \vy_j^{(t)}}$. Similarly to~\Cref{lemma:small-prob}, we will now show that $\hatvx_{i}^{(t)}, \hatvy_j^{(t)} \leq 2/\ell$ for any $(i, j) \in [n] \times [n]$, where $\ell \defeq k/T$ and $\gamma$ is sufficiently small. Indeed, by construction of $\cG'$,
\[
    1 + \frac{\gamma}{\ell} \leq \sw(\vmu) \leq (1 + \gamma) \pr_{(i , j) \sim \vmu } [ (i \leq n, j \leq n) \lor (i = \newaction, j = \newaction) ].
\]
Thus,
\[
     \pr_{(i, j ) \sim \vmu} [ (i \leq n, j \leq n) \lor (i = \newaction, j = \newaction) ] \geq 1 - \gamma.
\]
This implies that the expected utility of each player under $\vmu$ is at most $(1 + \gamma + k \gamma)/2$. Further, since $\vmu$ is also assumed to be a CCE, considering the deviation of Player $y$ to column $n + i$, we have that for any $i \in [n]$,
\[
    \frac{1 + \gamma + k \gamma}{2} \geq \frac{k}{2T} \vx_i^{(t)} - \frac{T-(1 - \gamma )}{T} r \implies \hatvx_i^{(t)} \leq \frac{2}{\ell}
\]
when $\gamma$ is sufficiently small, where we used the fact that $\sum_{i=1}^n \vx_i^{(t)} \geq p \geq 1 - \gamma$. Similar reasoning yields that $\hatvy_j^{(t)} \leq 2/\ell$ for any $j \in [n]$. Moreover, since $\sw(\vx^{(t)} \otimes \vy^{(t)}) \geq 1 + \gamma/\ell$, it is easy to see that $\sw(\hatvx^{(t)} \otimes \hatvy^{(t)}) \geq 1 + \gamma/\ell$, which in turn implies that $\sum_{i=1}^n \hatvx_i^{(t)} \hatvy_i^{(t)} \geq 1/\ell$. To conclude the argument, we will show---analogously to~\Cref{lemma:induced-clique}---how to extract an $\Omega(\ell)$-clique from that product distribution. Namely, we identify the set $S$ consisting of all $i \in [n]$ such that $\hatvx^{(t)}_i, \hatvy^{(t)}_i \geq 1/(16 \ell)$. Then,
\begin{align}
    \frac{1}{\ell} \leq \sum_{i=1}^n \hatvx_i^{(t)} \hatvy_i^{(t)} &= \sum_{i \in S} \hatvx_i^{(t)} \hatvy_i^{(t)} + \sum_{i \notin S} \hatvx_i^{(t)} \hatvy_i^{(t)} \notag \\
    &= \sum_{i \in S} \hatvx_i^{(t)} \hatvy_i^{(t)} + \sum_{i \notin S} \max(\hatvx_i^{(t)}, \hatvy_i^{(t)}) \min (\hatvx_i^{(t)}, \hatvy_i^{(t)}) \notag \\
    &\leq |S| \frac{4}{\ell^2} + \sqrt{ \sum_{i \notin S} \max(\hatvx_i^{(t)}, \hatvy_i^{(t)})^2 \sum_{i \notin S} \min ( \hatvx_i^{(t)}, \hatvy_i^{(t)})^2 } \label{align:bound-maxmin} \\
    &\leq |S| \frac{4}{\ell^2} + \sqrt{ \frac{4}{\ell} \frac{1}{16 \ell}},\label{align:claim-bound}
\end{align}
where~\eqref{align:bound-maxmin} uses the Cauchy-Schwarz inequality and the fact that $\hatvx_i^{(t)}, \hatvy^{(t)}_i \leq 2/\ell$, and~\eqref{align:claim-bound} is a consequence of~\Cref{claim:squared-bound}; to apply~\Cref{claim:squared-bound}, we note that $\sum_{i \notin S} \max(\hatvx_i^{(t)}, \hatvy_i^{(t)}) \leq \sum_{i \notin S} ( \hatvx_i^{(t)} + \hatvy_i^{(t)}) \leq 2$; $\sum_{i \notin S} \min(\hatvx_i^{(t)}, \hatvy_i^{(t)}) \leq \sum_{i \notin S} \hatvx_i^{(t)} \leq 1$; $\max(\hatvx_i^{(t)}, \hatvy_i^{(t)}) \leq 2/\ell$ for all $i \in [n]$; and $\min(\hatvx_i^{(t)}, \hatvy_i^{(t)}) \leq 1/(16\ell)$ for all $i \in [n] \setminus S$ (by definition of $S$). As a result, we conclude that $|S| \geq \ell/8$. Finally, similarly to the argument of~\Cref{lemma:induced-clique}, it is easy to see that $S$ induces a clique on $G$---the contrary case would contradict the fact that $\sw(\hatvx^{(t)} \otimes \hatvy^{(t)}) \geq 1 + \gamma/\ell$ when $\gamma$ is sufficiently small.

Therefore, under uniqueness we know that $G$ does not contain a clique of size $k$, while in the contrary case $G$ must contain a clique of size $\Omega(k/T)$. Taking $k = n^{1 - \epsilon}$ and $T = \Theta(n^{1 - 2 \epsilon})$ in conjunction with the hardness result of~\citet{Zuckerman07:Linear} completes the proof.
\end{proof}

It is not clear how to extend~\Cref{theorem:uniqueness} when we allow non-uniform mixtures. In particular, taking two mixtures on $(1, 1)$ and $(\newaction,\newaction)$ with weights $\alpha^{(1)} \approx 0$ and $\alpha^{(2)} \approx 1$, respectively, yields a ($2$-sparse) CCE for $\cG'$, but clearly does not contain any useful information about $G$. (Interestingly, the previously described correlated distribution is not a CE: a player can profitably deviate according to $1 \mapsto n+1$ and $\newaction \mapsto \newaction$.)


Next, the problem $\subsetsparsecce(\cG, T, S)$ asks whether there is a uniform $T$-sparse CCE supported solely on the joint actions given as input in $S$; this problem also does not hinge on an underlying objective. Using the previous reduction, by excluding the pair $(\newaction, \newaction)$, we obtain a similar \NP-hardness result.

\begin{theorem}
    \label{theorem:subset}
    $\subsetsparsecce(\cG, n^{1 - \epsilon}, S)$ is $\NP$-hard with respect to $n \times n$ games for any constant $\epsilon > 0$.
\end{theorem}

\subsection{Inapproximability and other objectives}
\label{sec:inapprox-obj}

We will next show that a small adjustment to the reduction of~\Cref{sec:further-impl} has two further important consequences: it enables capturing other natural objectives (\Cref{cor:egal-wel,cor:oneutil}), beyond welfare, and it precludes any multiplicative approximation with respect to the objective (\Cref{cor:inapprox}).

The driving force behind those results is~\Cref{theorem:basic-emb} below. For simplicity, in what follows we prove lower bounds concerning CCE (that is, with $\epsilon = 0$), although it is straightforward to extend the argument so as to account for some imprecision in the approximation of the equilibrium.

\begin{theorem}
    \label{theorem:basic-emb}
    Given as input an $n$-node graph $G$, $k \in [n]$ and $\epsilon \ll 1$, we can construct a $(2n+1) \times (2n+1)$ (two-player) game with the following properties:
    \begin{enumerate}
        \item it always admits a (pure) Nash equilibrium under which both players obtain a utility of $\epsilon$; \label{item:bad-Nash}
        \item when $G$ contains a clique of size $k$, there is a uniform $T$-sparse CCE such that each player obtains a utility of at least $1/2$; and \label{item:good-CCE}
        \item when $G$ does not contain a clique of size $k/(8T)$, there is no other uniform $T$-sparse CCE.\label{item:gap}
    \end{enumerate}
\end{theorem}

Assuming~\Cref{theorem:basic-emb}, we immediately establish \NP-hardness with respect to other natural objectives, besides the social welfare---any reasonable objective would refrain from selecting the former equilibrium (\Cref{item:bad-Nash}), if there is a choice to do so. In particular, we mention below two notable such implications.

\begin{corollary}
    \label{cor:oneutil}
    Computing a uniform $T$-sparse CCE that maximizes the utility of Player $x$ (or Player $y$) is $\NP$-hard with respect to $n \times n$ games for any $T \leq n^{1 - \epsilon}$ and constant $\epsilon > 0$.
\end{corollary}

For the next implication, we recall that the \emph{egalitarian} social welfare is the expected utility of the player who is worse---that objective is often advocated in the context of fairness.\footnote{Throughout this paper, we follow the convention that (social) welfare without any further specification refers to the utilitarian social welfare.}

\begin{corollary}
    \label{cor:egal-wel}
    Computing a uniform $T$-sparse CCE that maximizes the egalitarian social welfare is $\NP$-hard with respect to $n \times n$ games for any $T \leq n^{1 - \epsilon}$ and constant $\epsilon > 0$.
\end{corollary}

Moreover, \Cref{theorem:basic-emb} implies hardness of approximating the (utilitarian) welfare, measured here as a ratio. In particular, to make this problem meaningful, the reduction of~\Cref{theorem:basic-emb} makes sure that any joint action profile has nonnegative welfare.

\begin{corollary}
    \label{cor:inapprox}
    Computing a uniform $T$-sparse CCE that approximates the welfare-optimal one to any positive ratio is \NP-hard with respect to $n \times n$ games for any $T \leq n^{1 - \epsilon}$ and constant $\epsilon > 0$.
\end{corollary}

We mention in passing that~\Cref{theorem:basic-emb} resembles to a certain extent a result of~\citet{Conitzer08:New} concerning Nash equilibria. Their reduction is instead based on \textsc{SAT}, and it would be interesting to understand how it interacts with sparse CCE. Before we proceed with the proof of~\Cref{theorem:basic-emb}, we also point out that~\Cref{theorem:uniqueness,theorem:subset} are implied by~\Cref{theorem:basic-emb}, but the reduction in~\Cref{sec:further-impl} is easier to work with because deviating to $\newaction$ secures the same utility no matter how the other player acts (\Cref{fig:new-matrix}); this will no longer be the case for the reduction behind~\Cref{theorem:basic-emb}, complicating the argument.

\begin{proof}[Proof of~\Cref{theorem:basic-emb}]
The reduction is based on game $\cG'$ in the reduction of~\Cref{theorem:uniqueness}, with the only difference that the utility under $(\newaction, \newaction)$ is now taken to be $0 < \epsilon \ll 1$ for both players. As before, $(\newaction, \newaction)$ is a (pure) Nash equilibrium, and hence a $T$-sparse CCE (implying~\Cref{item:bad-Nash}). \Cref{item:good-CCE} also follows readily. It thus suffices to show that when $G$ does not contain a $k/(8T)$-clique, $(\newaction, \newaction)$ is the only uniform $T$-sparse CCE (\Cref{item:gap}).

Let $\vmu$ be a $T$-sparse CCE distinct from $(\newaction, \newaction)$. Considering the deviation of Player $x$ to $\newaction$ and of Player $y$ to $\newaction$, we have
\begin{align}
    \sw(\vmu) &\geq r \left( 1 - \frac{1}{T} \sum_{t=1}^T \vx_{\newaction}^{(t)} + 1 - \frac{1}{T} \sum_{t=1}^T \vy_{\newaction}^{(t)} \right) + \epsilon \left( \frac{1}{T} \sum_{t=1}^T \vx_{\newaction}^{(t)} + \frac{1}{T} \sum_{t=1}^T \vy_{\newaction}^{(t)} \right) \label{align:first-sw-lb} \\
    &\geq \frac{1}{T} \sum_{t=1}^T \left( (1 + \gamma) \left( 1 - \frac{ \vx^{(t)}_{\newaction} + \vy^{(t)}_{\newaction} }{2} \right) + 2 \epsilon \frac{ \vx^{(t)}_{\newaction} + \vy^{(t)}_{\newaction} }{2} \right) - \gamma \frac{1}{T} \sum_{t=1}^T \left( 1 - \frac{\vx^{(t)}_{\newaction} + \vy^{(t)}_{\newaction}}{2} \right). \label{align:sw-T-lb}
\end{align}
Further, since the game is zero-sum when excluding $(\newaction, \newaction)$ and the joint actions in the $n \times n$ submatrix corresponding to $\mat{A}$,
\begin{align}
    \sw(\vmu) &\leq 2 \epsilon \frac{1}{T} \sum_{t=1}^T \vx_{\newaction}^{(t)} \vy_{\newaction}^{(t)} + (1 + \gamma) \frac{1}{T} \sum_{t=1}^T \pr_{i \sim \vx^{(t)}, j \sim \vy^{(t)}} [i \leq n, j \leq n]  \notag \\
    &= \frac{1}{T} \sum_{t=1}^T \left( (1 + \gamma) ( (1 - \vx^{(t)}_{\newaction} ) ( 1 - \vy^{(t)}_{\newaction} ) - \delta^{(t)}) + 
2\epsilon \vx^{(t)}_{\newaction} \vy^{(t)}_{\newaction} \right),\label{align:sw-T-ub}
\end{align}
where we have defined $\delta^{(t)} \defeq \pr_{i \sim \vx^{(t)}, j \sim \vy^{(t)}} [i \leq 2n, j \leq 2n] - \pr_{i \sim \vx^{(t)}, j \sim \vy^{(t)}} [i \leq n, j \leq n]$ for each $t \in [T]$. Combining~\eqref{align:sw-T-lb} and~\eqref{align:sw-T-ub},
\begin{align*}
    \frac{1}{T} \sum_{t=1}^T \left( 2\epsilon \left( \frac{ \vx^{(t)}_{\newaction} + \vy^{(t)}_{\newaction} }{2} - \vx^{(t)}_{\newaction} \vy^{(t)}_{\newaction} \right)  + (1 + \gamma) \left( 1 - \frac{ \vx^{(t)}_{\newaction} + \vy^{(t)}_{\newaction} }{2} - (1 - \vx^{(t)}_{\newaction})(1 - \vy^{(t)}_{\newaction}) + \delta^{(t)} \right) \right) \\ \leq \gamma \frac{1}{T} \sum_{t=1}^T \left( 1 - \frac{\vx^{(t)}_{\newaction} + \vy^{(t)}_{\newaction}}{2} \right).
\end{align*}
Thus, for $\epsilon < \frac{1}{2}$,
\begin{align}
    \frac{1}{T} \sum_{t=1}^T \left( \frac{ \vx^{(t)}_{\newaction} + \vy^{(t)}_{\newaction} }{2} - \vx^{(t)}_{\newaction} \vy^{(t)}_{\newaction} + 
\frac{1}{2} \delta^{(t)} \right) &\leq \frac{\gamma}{2\epsilon + 1 + \gamma} \frac{1}{T} \sum_{t=1}^T \left( 1 - \frac{\vx^{(t)}_{\newaction} + \vy^{(t)}_{\newaction}}{2} \right) \notag \\
    &\leq \gamma \frac{1}{T} \sum_{t=1}^T \left( 1 - \frac{\vx^{(t)}_{\newaction} + \vy^{(t)}_{\newaction}}{2} \right).\notag 
\end{align}
Given that $\vx_{\newaction}^{(t)}, \vy^{(t)}_{\newaction} \in [0, 1]$ and $\delta^{(t)} \geq 0$, it follows that for any $t \in [T]$,
\[
    \frac{\vx_{\newaction}^{(t)} + \vy_{\newaction}^{(t)} }{2} - \vx^{(t)}_{\newaction} \vy^{(t)}_{\newaction} + \frac{1}{2} \delta^{(t)} \leq \gamma T \iff \delta^{(t)} + \vx^{(t)}_{\newaction} ( 1 - \vy^{(t)}_{\newaction} ) + \vy^{(t)}_{\newaction} (1 - \vx^{(t)}_{\newaction}) \leq 2 \gamma T.
\]
That is, we have shown that for any $t \in [T]$,
\begin{equation}
    \label{eq:diag-strong}
    \pr_{i \sim \vx^{(t)}, j \sim \vy^{(t)}} [ (i \leq n, j \leq n) \lor (i = \newaction, j = \newaction)] \geq 1 - 2 \gamma T. 
\end{equation}
Now, from the assumption that $\vmu$ is not supported only on $(\newaction, \newaction)$, it follows that there exists $t \in [T]$ such that $\vx^{(t)}_{\newaction} \neq 1 \lor  \vy^{(t)}_{\newaction} \neq 1$. In particular, by~\eqref{align:first-sw-lb}, there exists such $t$ with the property that
\begin{equation}
    \label{eq:T-ineq}
    2r \left( 1 - \frac{\vx^{(t)}_{\newaction} + \vy^{(t)}_{\newaction}}{2} \right) + 2\epsilon \frac{\vx^{(t)}_{\newaction} + \vy^{(t)}_{\newaction}}{2} \leq 2 \epsilon \vx^{(t)}_{\newaction} \vy^{(t)}_{\newaction} + \pr_{i \sim \vx^{(t)}, j \sim \vy^{(t)}}[i \leq n, j \leq n] \sw(\hatvx^{(t)} \otimes \hatvy^{(t)}),
\end{equation}
where again we have defined $\Delta^n \ni \hatvx_{i \leq n}^{(t)} =  \nicefrac{\vx^{(t)}_{i \leq n}}{\sum_{i=1}^n \vx_i^{(t)}}$ and $\Delta^n \ni \hatvy_{j \leq n}^{(t)} =  \nicefrac{\vy^{(t)}_{j \leq n}}{\sum_{j=1}^n \vy_j^{(t)}}$; any $t$ such that $\vx^{(t)}_{\newaction} \neq 1 \lor  \vy^{(t)}_{\newaction} \neq 1$ that satisfies~\eqref{eq:T-ineq} must also satisfy $\sum_{i=1}^n \vx_i^{(t)} > 0$ and $\sum_{j=1}^n \vy_j^{(t)} > 0$. Given that $( \vx_{\newaction}^{(t)} + \vy^{(t)}_{\newaction} )/2 \geq \vx^{(t)}_{\newaction} \vy^{(t)}_{\newaction}$, we have
\[
    2r \left( 1 - \frac{\vx^{(t)}_{\newaction} + \vy^{(t)}_{\newaction}}{2} \right) \leq \pr_{i \sim \vx^{(t)}, j \sim \vy^{(t)}}[i \leq n, j \leq n] \sw(\hatvx^{(t)} \otimes \hatvy^{(t)}).
\]
Further,
\[
  \pr_{i \sim \vx^{(t)}, j \sim \vy^{(t)}}[i \leq n, j \leq n] \leq ( 1 - \vx_{\newaction}^{(t)} )  ( 1 - \vy_{\newaction}^{(t)}) \leq \left( 1 - \frac{\vx_{\newaction}^{(t)} + \vy_{\newaction}^{(t)} }{2} \right) \neq 0.
\]
We thus conclude that $\sw(\hatvx^{(t)} \otimes \hatvy^{(t)}) \geq 2r = 1 + \gamma T/k$. We will now argue that $\hatvx_i^{(t)}, \hatvy_j^{(t)} \leq 2/\ell$ for any $i, j \in [n]$. Continuing from~\eqref{eq:T-ineq}, we bound
\[
    \left( 1 + \gamma \frac{T}{k} \right) \left( 1 - \frac{\vx^{(t)}_{\newaction} + \vy^{(t)}_{\newaction}}{2} \right) + 2\epsilon \frac{\vx^{(t)}_{\newaction} + \vy^{(t)}_{\newaction}}{2} \leq 2 \epsilon \vx^{(t)}_{\newaction} \vy^{(t)}_{\newaction} + ( 1 + \gamma) (1 - \vx^{(t)}_{\newaction} ) (1 - \vy^{(t)}_{\newaction}).
\]
Thus,
\[
    (1 + \gamma) \frac{\vx_{\newaction}^{(t)} + \vy_{\newaction}^{(t)} }{2} - \vx^{(t)}_{\newaction} \vy^{(t)}_{\newaction} \leq \gamma,
\]
By applying the AM-GM inequality, we have
\[
    (1 + \gamma) \sqrt{ \vx_{\newaction}^{(t)} \vy_{\newaction}^{(t)} } - \vx_{\newaction}^{(t)} \vy^{(t)}_{\newaction} \leq \gamma \iff \left( 1- \sqrt{ \vx_{\newaction}^{(t)} \vy_{\newaction}^{(t)} } \right) \left( \sqrt{ \vx_{\newaction}^{(t)} \vy_{\newaction}^{(t)} } - \gamma \right) \leq 0 \iff \sqrt{ \vx_{\newaction}^{(t)} \vy_{\newaction}^{(t)} } \leq \gamma,
\]
where we used the fact that $\vx^{(t)}_{\newaction} \neq 1 \lor  \vy^{(t)}_{\newaction} \neq 1$. Combining with~\eqref{eq:diag-strong}, it follows that $\pr_{i \sim \vx^{(t)}, j \sim \vy^{(t)}} [ i \leq n, j \leq n] \geq 1 - \gamma^2 - 2 \gamma T$. This implies that $\sum_{i=1}^n \vx_i^{(t)} \geq 1 - \gamma^2 - 2 \gamma T$ and $\sum_{j=1}^n \vy_j^{(t)} \geq 1 - \gamma^2 - 2 \gamma T$. Again by~\eqref{eq:diag-strong}, it follows that the expected utility of each player under $\vmu$ is at most $ (1 + \gamma) /2 + \gamma k T$. Moreover, since $\vmu$ is assumed to be a CCE, by considering the deviation of Player $y$ to column $n + i$, we have that for any $i \in [n]$,
\[
    \frac{1 + \gamma + 2 \gamma k T }{2} \geq \frac{k}{2} \frac{1}{T} \sum_{t=1}^T \vx_i^{(t)} - r \frac{1}{T} \sum_{t=1}^T \vx^{(t)}_{\newaction} \geq \frac{k}{2T} \vx_i^{(t)} - \frac{T - (1 - \gamma^2 - 2 \gamma T)}{T} r \implies \hatvx_i^{(t)} \leq \frac{2 T}{k},
\]
when $\gamma$ is sufficiently small. Finally, we let $S \defeq \{i \in [n]: \hatvx^{(t)}_i, \hatvy^{(t)}_i \geq 1/(16 \ell) \}$, where $\ell \defeq k/T$. As in the proof of~\Cref{theorem:uniqueness}, it follows that $|S| \geq \ell/8$, and $S$ is a clique on $G$. As a result, when $G$ does not contain a clique of size $k/(8T)$, $(\newaction, \newaction)$ is the only uniform $T$-sparse CCE. This completes the proof.
\end{proof}

An interesting question, which is left for future work, is whether the construction of~\Cref{theorem:basic-emb} can be extended to hold for non-uniform sparse CCE.

\section{Hardness for low-precision equilibria}
\label{sec:low-precision}

So far, we have established lower bounds for $\sparsecce$ in the regime where both the equilibrium and the optimality gap are $\poly (1/n)$. No-regret learning is often employed in the low-precision regime, where the precision scales as $1/\polylog n$. We will show here how to obtain hardness results in the latter regime, establishing~\Cref{theorem:main2} presented earlier in the introduction.

In that setting, we first remark that, even for $1$-sparse CCE (that is, Nash equilibria), \citet{Lipton03:Playing} famously gave a quasipolynomial-time algorithm based on enumerating all possible strategies of support $O(\log n/\epsilon^2)$. It is thus unlikely---under our current understanding of complexity---that the induced problem is \NP-hard. Instead, we will rely on~\Cref{conj:planted} pertaining to the planted clique problem. 

As before, our first lower bound (\Cref{theorem:plantedclique}) makes use of~\Cref{alg:maxclique}, but with $G$ now being the input of the planted clique problem (presented in~\Cref{sec:max-planted}). We will use the following lemma, which was observed in the form below by~\citet{Hazan11:How}.

\begin{lemma}[\nakedcite{Hazan11:How}]
    \label{lemma:smallk}
    Suppose that there is randomized polynomial-time algorithm that given a planted clique problem with $k \geq c \log n$, for a constant $c > 0$, finds a clique of size $100 \log n$ with probability at least $1/2$. Then, there is a randomized polynomial-time algorithm that solves the planted clique problem for any $k \geq c_0 \log n$, for a constant $c_0 > 0$, with high probability.
\end{lemma}

Consider now any sparsity parameter $T = \polylog n$, and define $k = 100 T \log n$. By the guarantee given in~\Cref{subsec:prooftheorem} concerning~\Cref{alg:maxclique}, we know that we can then find a clique of size $100 \log n$ by invoking an oracle to $\sparsecce(\cG, T, (\log n)^{-c}, (\log n)^{-c})$, where $\cG = \cG(G)$ is defined as in~\Cref{line:game} and $c$ is a sufficiently large constant (that depends on the sparsity). In turn, \Cref{lemma:smallk} above implies that we can solve the planted clique problem (w.h.p.). We thus arrive at the following.

\begin{theorem}
    \label{theorem:plantedclique}
    Assuming that~\Cref{conj:planted} holds, $\sparsecce(\cG, T, (\log n)^{-c}, (\log n)^{-c})$ requires time $n^{\Omega(\log n)}$ for any sparsity $T = \polylog n$ and some constant $c = c(T)$.
\end{theorem}

Finally, we show hardness in the regime where the equilibrium and the optimality gap are both constants using a different approach. Namely, we extend the reduction of~\citet{Hazan11:How}, which was developed with Nash equilibria in mind. The main result is recalled below.

\begin{restatable}{theorem}{constant}
    \label{theorem:constant}
    Assuming that~\Cref{conj:planted} holds, $\sparsecce(\cG, T, c, c)$ requires time $n^{\Omega(\log n)}$ for any sparsity $T = O(1)$ and some constant $c = c(T)$.
\end{restatable}

We dedicate the rest of this section to the proof of~\Cref{theorem:constant}.

\subsection{Proof of Theorem~\ref{theorem:constant}}

The proof mostly follows the argument of~\citet{Hazan11:How}, but with certain modifications. Similarly to~\eqref{eq:game}, we consider a two-player game defined by the payoff matrices
\begin{equation}
    \label{eq:Hazangame}
    \R^{N \times N} \ni \mat{R} \defeq 
    \frac{1}{2}  
    \begin{pmatrix}
       \mat{A} & - \mat{B}^\top \\
        \mat{B} & \mat{0}_{(N - n) \times (N - n)}
    \end{pmatrix}
    \text{ and }
    \R^{N \times N} \ni \mat{C} \defeq
    \frac{1}{2} 
    \begin{pmatrix}
       \mat{A} & \mat{B}^\top \\
        - \mat{B} & \mat{0}_{(N-n) \times (N-n)}
    \end{pmatrix},
\end{equation}
where $N \gg n$ is sufficiently large (to be specified in the sequel), $\mat{A}$ is again the \emph{adjacency} matrix (with $1$ on the diagonal entries) of $G$---the input of the planted clique problem, and $\mat{B} \in \R^{(N - n) \times n}$ is defined by selecting each entry independently as
\[
    \mat{B}_{i, j} =
    \begin{cases}
        M & \text{with probability } \frac{3}{4} \frac{1}{M}, \text{ and }\\
        0 & \text{otherwise}.
    \end{cases}
\]
Here, $M$ is a sufficiently large parameter that will be related to the sparsity (see~\Cref{lemma:bipartite}). In what follows, we recall that $k \in \N$ represents the size of the planted clique. We first note that $(\mat{R}, \mat{C})$ in~\eqref{eq:Hazangame} admits, with high probability, a Nash equilibrium with welfare $1$. This proof of completeness is similar to the one by~\citet{Hazan11:How}.

\begin{lemma}
    Consider $(\mat{R}, \mat{C})$ with $k \geq 96 M^2 \log N $. Then, with high probability, there is a Nash equilibrium with welfare $1$.
\end{lemma}

\begin{proof}
    We consider the strategy profile $(\vx, \vy)$ in which both players play uniformly at random over the nodes corresponding to the planted clique. It is clear that $\sw(\vx \otimes \vy) = 1$. It thus suffices to show that $(\vx, \vy)$ is indeed a Nash equilibrium of~\eqref{eq:Hazangame}. By symmetry, we can analyze only deviations of Player $x$. When $i \leq n$, it is clear that there is no profitable deviation for Player $x$ as that player obtains $1/2$ under $(\vx, \vy)$, which is the highest utility attainable in the $n \times n$ submatrix corresponding to $\mat{A}$. Let us consider thus some $i > n$. The obtained utility then is half of the average over $k$ entries of $\mat{B}$, which is highly concentrated around its the mean value. In particular, we will use the following Chernoff–Hoeffding bound~\citep{Hoeffding63:Probability}. Let $X_1, \dots, X_k$ be independent random variables such that $|X_1|, \dots, |X_k| \leq C$ (almost surely). If $\bar{X}$ denotes their average, it holds that $\pr [ \bar{X} - \E[ \bar{X} ] \geq t] \leq e^{ - \frac{k t^2}{2C^2} }$ for any $t > 0$. In our case, if we set $C = M$ and $t = 1/4$, we find that the probability that $i$ is profitable is upper bounded by $e^{- \frac{k}{32 M^2}}$. Thus, by a union bound, the probability that there is a profitable deviation among $i > n$ is upper bounded by $(N-n) e^{- \frac{k}{32 M^2}} \leq 1/n^2$, where we used the assumption that $k \geq 96 M^2 \log N$. This completes the proof.
\end{proof}

We proceed by establishing soundness. The following lemma mirrors~\Cref{lemma:bad-mass} that was shown earlier, and follows from the same observation.

\begin{lemma}
    \label{lemma:bad-mass-new}
    Given a $T$-sparse $\epseq$-CCE of $(\mat{R}, \mat{C})$ with welfare $1 - \epswel$, we can compute in polynomial time an $(\epseq+ 2M \epswel)$-CCE of $(\mat{R}, \mat{C})$ with welfare at least $1 - \epswel$ and supported only on $[n] \times [n]$.
\end{lemma}

The proof of this lemma is similar to that of~\Cref{lemma:bad-mass}. In particular, if we define $\delta^{(t)} \defeq 1 - \sum_{i=1}^n \vx_i^{(t)} \sum_{j=1}^n \vy_j^{(t)}$ for each $t \in [T]$, we have
\[
1 - \epswel \leq \sum_{t=1}^T \alpha^{(t)} \sw(\vx^{(t)} \otimes \vy^{(t)}) = \sum_{t=1}^T \alpha^{(t)} \sum_{i=1}^n \sum_{j=1}^n \vx_i^{(t)} \mat{A}_{i, j} \vy_j^{(t)} \leq 1 - \sum_{t=1}^T \alpha^{(t)} \delta^{(t)},
\]
in turn implying that $\sum_{t=1}^T \alpha^{(t)} \delta^{(t)} \leq \epswel$. The rest of the argument then is the same as in~\Cref{lemma:bad-mass}, giving the conclusion of~\Cref{lemma:bad-mass-new}.

Next, the lemma below gives a bound analogous to~\Cref{lemma:small-prob}. (In what follows, we can safely assume that $c_2 \log n$ is an integer.)

\begin{lemma}
    \label{lemma:near-uni}
    Let $\vmu = \sum_{t=1}^T \alpha^{(t)} (\vx^{(t)} \otimes \vy^{(t)})$
    be a $\frac{1}{2}$-CCE of $(\mat{R}, \mat{C})$ such that each $\vx^{(t)} \otimes \vy^{(t)}$ is supported only on $[n] \times [n]$. Suppose further that $N = n^{c_1} \geq 2n$ and $c_1  - c_2 \log (4M/3) \geq 2$. Then, with high probability, the overall probability mass every player places on every subset of $c_2 \log n$ pure strategies of the $t$th product distribution is at most $2/(M \alpha^{(t)})$.
\end{lemma}

\begin{proof}
    Let $d \defeq c_2 \log n$. For the sake of contradiction, suppose that Player $y$ places on some set of $c_2 \log n$ pure strategies of the $t$th product distribution more than $2/(M \alpha^{(t)})$ probability mass. We will compute the probability that there is a row of $\mat{B}$ in which all entries corresponding to $d$ have a value of $M$. If that event happens, Player $x$ can deviate to that row and obtain a utility larger than $1$, while the current utility is at most $1/2$. The probability that event occurs for a single row is $p^d$ for $p \defeq \frac{3}{4M}$, but $\mat{B}$ has $N-n$ rows each with independent randomization (by construction of $\mat{B}$). Thus, the probability that no such rows exists can be bounded as
    \[
        (1 - p^d)^{N - n} \leq e^{ - p^d (N - n) } = e^{ - (N - n) n^{ - c_2 \log (4M/3)} } \leq e^{-\frac{n^2}{2}},
    \]
    where we used the assumption that $N = n^{c_1} \geq 2n$ and $c_1 - c_2 \log(4M/3) \geq 2$. Finally, we can rule out all subsets with size $c_2 \log n$ via a union bound, leading to the claim.
\end{proof}

In the lemma below, the \emph{density} of sets $S, T \subseteq [n]$ with $|S| = |T|$, denoted by $\dens(S, T)$, is the probability that an edge exists (with self-loops included) between two nodes chosen uniformly at random from $S$ and $T$.

\begin{lemma}
    \label{lemma:bipartite}
    Let $M > 4 T$ and $\epswel \leq 1/(10T)$. Given a $\frac{1}{2}$-CCE of $(\mat{R}, \mat{C})$ supported only on $[n] \times [n]$ with welfare at least $1 - \epswel$, we can find in polynomial time sets $S, T \subseteq [n]$ with $ |S| = |T| = c_2 \log n$ and $\dens(S, T) \geq \frac{3}{5}$.
\end{lemma}

\begin{proof}
    Let $\vmu = \sum_{t=1}^T \alpha^{(t)} (\vx^{(t)} \otimes \vy^{(t)})$, where with a slight abuse of notation we let $\vx^{(t)}, \vy^{(t)} \in \Delta^n$ for all $t \in [T]$ (since $\vmu$ is assumed to be supported only on $[n] \times [n]$). We select any $t \in [T]$ such that $\alpha^{(t)} \geq 1/T$. Since it is assumed that $\vmu$ has welfare at least $1 - \epswel$, it follows that
    \[
        \langle \vx^{(t)}, \mat{A} \vy^{(t)} \rangle \geq \frac{1}{\alpha^{(t)}} \left( \sw(\vmu) - \sum_{\tau \neq t} \alpha^{(\tau)} \sw(\vx^{(\tau)} \otimes \vy^{(\tau)}) \right) \geq 1 - \delta,
    \]
    where $\delta = \epswel T$. Let $I_y^{(t)} \defeq \{ j : \supp(\vy^{(t)}) : \langle \vx^{(t)}, \mat{A} \vec{e}_j \rangle  \geq \frac{4}{5} \}$, where $\vec{e}_j$ is the $j$th unit vector. Then, $\sum_{j \in I_y^{(t)}} \vy_j^{(t)} \geq 1 - 5 \delta$, for otherwise $\langle \vx^{(t)}, \mat{A} \vy^{(t)} \rangle < (1 - 5 \delta) + (5 \delta) \frac{4}{5} = 1 - \delta$. Using the fact that $\epswel \leq \frac{1}{10T} < \frac{1}{5T} - \frac{2}{5M}$ (since $M > 4 T$), it follows that $\sum_{j \in I_y^{(t)}} \vy_j^{(t)} > \frac{2T}{M} \geq \frac{2}{M\alpha^{(t)}}$. As a result, \Cref{lemma:near-uni} tells us that 
    $|I_y^{(t)}| \geq c_2 \log n$. If $|I_y^{(t)}| > c_2 \log n$, we let $I_y^{(t)}$ be an arbitrary subset of $I_y^{(t)}$ with size $c_2 \log n$. We now consider $\uni(I_y^{(t)}) \in \Delta^n$, the uniform distribution over $I_y^{(t)}$. We have
    \[
        \langle \vx^{(t)}, \mat{A} \uni(I_y^{(t)}) \rangle \geq \frac{4}{5}.
    \]
    Similarly, we define $I_x^{(t)} \defeq \{ i \in \supp(\vx^{(t)}) : \langle \vec{e}_i, \mat{A} \vec{u}(I_y^{(t)}) \rangle \geq \frac{3}{5} \}$. Then, $\sum_{ i \in I_x^{(t)} } \vx_i^{(t)} \geq \frac{1}{2} > \frac{2 T}{M}$ (since $M > 4 T$). Thus, \Cref{lemma:optimalCCE} implies that $|I_x^{(t)}| \geq c_2 \log n$. If $|I_x^{(t)}| > c_2 \log n$, we let $I_x^{(t)}$ be an arbitrary subset of $I_x^{(t)}$ with size $c_2 \log n$. Thus, we have found sets $I_x^{(t)}, I_y^{(t)} \subseteq [n]$ with $|I_x^{(t)}| = |I_y^{(t)}| = c_2 \log n$ and
    \[
        \langle \vec{u}(I_x^{(t)}), \mat{A} \vec{u}(I_y^{(t)}) \rangle \geq \frac{3}{5};
    \]
    that is, $\dens(I_x^{(t)}, I_y^{(t)}) \geq \frac{3}{5}$.
\end{proof}

We finally conclude the proof of soundness by stating a lemma shown by~\citet{Hazan11:How}, which we include without proof.

\begin{lemma}[\nakedcite{Hazan11:How}]
    \label{lemma:Hazan-dens}
    Given sets $S, T \subseteq [n]$ with $|S| = |T| = 2000 \log n$ and $\dens(S, T) \geq \frac{3}{5}$, we can find in polynomial time a clique of size $100 \log n$ with high probability.
\end{lemma}

As a result, \Cref{theorem:constant} follows from~\Cref{lemma:Hazan-dens,lemma:bipartite,lemma:near-uni,lemma:bad-mass-new,lemma:smallk} by making the following choice of parameters:

\begin{itemize}
    \item $M \defeq 5 T$;
    \item $c_2 \defeq 2000$;
    \item $c_1 \defeq \lceil c_2 \log (4M/3) \rceil + 2$ and $N \defeq n^{c_1}$;
    \item $k \defeq \lceil 96 M^2 \log N \rceil$;
    \item $\epswel \defeq \epseq/(4M)$; and
    \item $\epsilon \defeq 1/4$.
\end{itemize}

This reduction can also be used when $T = \omega(1)$, but then the required precision for $\epsilon$ and $\epswel$ would need to be accordingly small (we recall that one has to ultimately divide the entries of~\eqref{eq:Hazangame} by $M = \Theta(T)$ so as to normalize the utilities in $[-1, 1]$, which would result in $\epsilon = \Theta(1/T)$). It is not clear how to overcome this issue and prove stronger lower bounds with the current approach, but we suspect that such lower bounds should apply.
\section{Future research}

Our results raise a number of interesting avenues for future research. First, does~\Cref{theorem:basic-emb} apply with respect to \emph{non-uniform} sparse CCE? That would readily generalize~\Cref{cor:egal-wel,cor:oneutil,cor:inapprox}. Turning to the low-precision regime, what is the precise threshold of computational tractability for $\sparsecce$? \Cref{theorem:constant} only precludes $T = O(1)$ (under constant precision), but we suspect that it can be significantly strengthened. Furthermore, it seems likely that one can instead prove lower bounds by relying on the exponential-time hypothesis (ETH)---as opposed to the planted clique conjecture (\emph{cf.}~\citet{Braverman15:Approximating}), which would provide further evidence for the intractability of the problem. Yet, extending existing ETH-based reductions to sparse CCE appears to introduce considerable challenges. Last but not least, perhaps the most important question concerns the complexity of computing sparse CCE without imposing further constraints, a problem that will likely require quite different techniques from the ones employed here.



\section*{Acknowledgments}

This material is based on work supported by the Vannevar Bush Faculty Fellowship ONR N00014-23-1-2876, National Science Foundation grants RI-2312342 and RI-1901403, ARO award W911NF2210266, and NIH award A240108S001. Alkis Kalavasis was supported by the Institute for Foundations of Data Science at Yale.

\printbibliography

\end{document}